\documentclass{article}
\usepackage{graphicx}
\usepackage{hyperref}
\usepackage{lineno}
\usepackage[OT1]{fontenc}
\usepackage{amsthm,amsmath,amssymb}
\usepackage[numbers]{natbib}
\usepackage{graphicx}
\usepackage{subfig}
\usepackage{titlesec}

\RequirePackage{enumitem} 
\RequirePackage{mathtools}
\usepackage[doublespacing]{setspace}
\usepackage{geometry}
\def\T{{ \mathrm{\scriptscriptstyle T} }}
\numberwithin{equation}{section}
\theoremstyle{plain}
\newtheorem{thm}{Theorem}[section]
\newtheorem{lemma}[thm]{Lemma}
\titlelabel{\thetitle.\quad}
\def\bibindent{1em}

\makeatletter
\renewcommand\@biblabel[1]{} 
\renewenvironment{thebibliography}[1]
{\section*{\refname}%
	\@mkboth{\MakeUppercase\refname}{\MakeUppercase\refname}%
	\list{\@biblabel{\@arabic\c@enumiv}}%
	{\settowidth\labelwidth{\@biblabel{}}%
		\leftmargin\labelwidth
		\advance\leftmargin15pt
		\advance\leftmargin\labelsep
		\setlength\itemindent{-10pt}
		\@openbib@code
		\usecounter{enumiv}%
		\let\p@enumiv\@empty
		\renewcommand\theenumiv{\@arabic\c@enumiv}}%
	\sloppy
	\clubpenalty4000
	\@clubpenalty \clubpenalty
	\widowpenalty4000%
	\sfcode`\.\@m}
{\def\@noitemerr
	{\@latex@warning{Empty `thebibliography' environment}}%
	\endlist}
\renewcommand\newblock{\hskip .11em\@plus.33em\@minus.07em}
\makeatother

\newenvironment{bottompar}{\par\vspace*{\fill}}{\clearpage}
\renewcommand{\refname}{REFERENCES}
\begin{document}

\title{\textbf{Group Sequential Crossover Trial Designs with Strong Control of the Familywise Error Rate}}
\author{\textbf{M. J. Grayling\textsuperscript{1}, J. M. S. Wason\textsuperscript{1,2}, A. P. Mander\textsuperscript{1}}\\
	\small 1. Hub for Trials Methodology Research, MRC Biostatistics Unit, Cambridge, UK, \\ \small 2. Institute of Health and Society, Newcastle University, Newcastle, UK.}
\date{}
\maketitle

\noindent \textbf{Running Head:} Group Sequential Crossover Trials.\\

\noindent \textbf{Abstract:} Crossover designs are an extremely useful tool to investigators, whilst group sequential methods have proven highly proficient at improving the efficiency of parallel group trials. Yet, group sequential methods and crossover designs have rarely been paired together. One possible explanation for this could be the absence of a  formal proof of how to strongly control the familywise error rate in the case when multiple comparisons will be made. Here, we provide this proof, valid for any number of initial experimental treatments and any number of stages, when results are analysed using a linear mixed model. We then establish formulae for the expected sample size and expected number of observations of such a trial, given any choice of stopping boundaries. Finally, utilising the four-treatment, four-period TOMADO trial as an example, we demonstrate group sequential methods in this setting could have reduced the trials expected number of observations under the global null hypothesis by over 33\%.\\



\begin{bottompar}
	\noindent Address correspondence to M. J. Grayling, MRC Biostatistics Unit, Forvie Site, Robinson Way, Cambridge CB2 0SR, UK; Fax: +44-(0)1223-330365; E-mail: mjg211@cam.ac.uk. 
\end{bottompar}

\section{INTRODUCTION}
\label{s:intro}

The efficiency of crossover trials often makes them the best design for a clinical trial. Administering multiple treatments to patients reduces the standard error of the estimated treatment effects compared to a parallel trial design with an equal number of patients. Therefore, whilst restrictions to their use exist; such as a requirement for patients to begin each new treatment period in a comparable state to those completed, crossover trials are the design of choice in many settings (Jones and Kenward, 2014; Senn, 2002), resulting in them accounting for 22\% of all published trials in December 2000 for example (Mills et al., 2009).

In a parallel design setting, group sequential methods are frequently utilized to improve a clinical trials efficiency (Jennison and Turnbull, 2000). These designs incorporate interim analyses which allow for early rejection of null hypotheses; efficacy stopping, or early stopping for lack of benefit; futility stopping. This way, the expected sample size required can be reduced over the more classical single-stage approach. Moreover, multi-arm multi-stage designs, which allow multiple experimental treatments to share a control group, can increase efficiency even further (Parmar et al., 2014).

Group sequential methods are not frequently used in crossover trial settings however, in particular ones with multiple experimental treatments. Hauck et al. (1997) investigated the performance of group sequential trials for average bioequivalence employing an AB/BA crossover design, whilst Jennison and Turnbull (2000) provided one possible analysis method for a group sequential AB/BA crossover with a normally distributed endpoint. No one, to the best of our knowledge, has explored group sequential theory for crossover trials with more than one experimental treatment being compared to a shared control.

Thus, one possible explanation for the lack of group sequential crossover trials may be that there is not yet available a formal proof of how to strongly control the familywise error rate of such a trial with multiple experimental treatments; since such a proof is usually required for regulatory approval (Wason et al., 2014). In comparison to a proof for a parallel multi-arm multi-stage design (Magirr et al., 2012), proving strong control of the familywise error rate is complicated here due to difficulties associated with the covariance structure implied by mixed model analysis. As has been remarked, multiple testing corrections for mixed models are only presently available for certain specific circumstances (Bender and Lange, 2001). Extension to this setting is particularly significant though given the noted advantages of comparing multiple experimental treatments to a shared control, both in terms of trial management and sample size (Parmar et al., 2014).

Potential exists, given a proof, for the efficiency of crossover trial designs to be improved. In this work, we begin by providing such a proof for a linear mixed model using period and treatment as fixed effects, and individuals as random effects. Following this, using the four treatment, four-period TOMADO trial (Quinnell et al., 2014) as an example, we explore and discuss the efficiency gains that group sequential designs could bring in a crossover setting.

\section{METHODS}
\label{s:model}

\subsection{Notation, Hypotheses and Analysis}

The trial is assumed to have \(D \ge 2\) treatments initially, indexed \(d=0,\dots,D-1\). Treatments \(d=1,\dots,D-1\) are experimental, to be compared to the control \(d=0\). A maximum of \(L\) stages are planned for the trial. At each stage patients are allocated to each of a set of treatment sequences, which specify an order in which a patient receives treatments. The sequences used at each stage are determined by the number of treatments remaining in the trial at that stage. Without loss of generality, we will assume that if a treatment, or treatments, are dropped, it is treatment \(D-1\) dropped first, then \(D-2\), and so on, since treatments can always be re-labelled at each interim analysis. Then, we denote by \(S_r=\{s_{ri} : i=1,\dots,\left|S_r\right|\}\), \(r = 2,\dots,D\), the set of sequences for patient treatment allocation when $r$ treatments remain in the trial, with each $S_r$ written in the form assuming it is exactly treatments $d = 0,\dots, r-1$ that remain. We further constrain each \(S_r\) to contain only complete block sequences that are balanced for period. Specifically, complete block allocation requires all sequences to contain each treatment remaining in the trial exactly once, and period balance requires an equal number of patients to receive each treatment remaining in the trial in each period. These constraints allow the use of the popular Latin and Williams squares (Jones and Kenward, 2014).

A fixed group size \(n\) is used for each stage of the trial, and is chosen such that at every stage each sequence is used an equal number of times. Thus \(n\) must be divisible by the lowest common multiple of \(\left|S_2\right|,\dots,\left|S_D\right|\). Designing the trial in this manner ensures each treatment is considered equally.

Outcome data is assumed to be normally distributed, and a linear mixed model is used for analysis, given by
\[y_{ijkl} = \mu_0 + \pi_j + \tau_{d[j,k,l]} + s_{ikl} + \epsilon_{ijkl},\]
or
\[\boldsymbol{Y} = \boldsymbol{X}\boldsymbol{\beta} + \boldsymbol{Z}\boldsymbol{b} + \boldsymbol{\epsilon},\]
where
\begin{itemize}
	\item \(\boldsymbol{Y}\) is the vector of responses, containing the values of the \(y_{ijkl}\); the response for individual \(i\), in period \(j\), on sequence \(k\), in stage \(l\),
	\item \(\boldsymbol{\beta}\) is the vector of fixed effects, of length \(2D-1\), consisting of
	\begin{itemize}
		\item \(\mu_0\) the mean response on treatment 0 in period 1, an intercept term,
		\item \(\pi_j\) the fixed period effect for period \(j\), with the identifiability constraint \(\pi_1=0\). Note that period is reset to 1 for each new stage of the trial. That is, the first period of stage 2 is treated as period 1 rather than period $D+1$, and similarly for later stages. Thus, we have exactly \(D-1\) non-zero period effects given our restriction to complete block sequences.
		\item \(\tau_{d[j,k,l]}\)  is the fixed direct treatment effect for an individual in period \(j\), on sequence \(k\), in stage \(l\), with the identifiability constraint \(\tau_0=0\),
	\end{itemize}
	\item \(\boldsymbol{X}\) is the matrix linking the fixed effects to the vector of responses,
	\item \(\boldsymbol{b}\) is the vector of random effects, consisting of the \(s_{ikl}\); the random effect for individual \(i\), on sequence \(k\), in stage \(l\),
	\item \(\boldsymbol{Z}\) is the matrix linking the random effects to the vector of responses,
	\item \(\boldsymbol{\epsilon}\) is the vector of residuals, consisting of the \(\epsilon_{ijkl}\); the residual for individual \(i\), in period \(j\), on sequence \(k\), in stage \(l\).
\end{itemize}
Additionally, denoting by \(\sigma_b^2>0\) and \(\sigma_e^2>0\) the between and within subject variances respectively, we take
\begin{align*}
\text{cov}(s_{i_1k_1l_1},s_{i_2k_2l_2}) &= \sigma_b^2 \delta_{i_1i_2} \delta_{k_1k_2} \delta_{l_1l_2},\\ \text{cov}(\epsilon_{i_1j_1k_1l_1},\epsilon_{i_2j_2k_2l_2}) &= \sigma_e^2 \delta_{i_1i_2} \delta_{j_1j_2} \delta_{k_1k_2} \delta_{l_1l_2},
\end{align*}
where \(\delta_{ij}\) is the Kronecker Delta function. Incorporation of fixed effects for period and treatment only, and our chosen covariance structure above, are the conventional choices for a crossover trial (Jones and Kenward, 2014).

We test \(D-1\) hypotheses. Since we are interested in testing the efficacy of experimental treatments in comparison to a control, we consider the case of one-sided alternative hypotheses $H_{0d} : \tau_d \le 0, H_{1d} : \tau_d > 0,$ for  \(d = 1,\dots,D-1\).

At each interim analysis the above model is used to compute an estimate, \(\hat{\boldsymbol{\beta}}_l\) \((l=1,\dots,L)\), for \(\boldsymbol{\beta}\) through the standard maximum likelihood estimator of a linear mixed model
\[\hat{\boldsymbol{\beta}}_l = (\boldsymbol{X}^\T \boldsymbol{\Sigma}^{-1} \boldsymbol{X} )^{-1}\boldsymbol{X}^\T\boldsymbol{\Sigma}^{-1}\boldsymbol{Y} \sim MVN\left\{\boldsymbol{\beta}, (\boldsymbol{X}^\T \boldsymbol{\Sigma}^{-1} \boldsymbol{X} )^{-1}\right\},\]
where $\boldsymbol{\Sigma} = \boldsymbol{Z}\text{cov}(\boldsymbol{b},\boldsymbol{b})\boldsymbol{Z}^\T + \text{cov}(\boldsymbol{\epsilon},\boldsymbol{\epsilon})$ (Fitzmaurice et al., 2011). From this we acquire \(\hat{\boldsymbol{\tau}}_l = (\hat{\tau}_{1l},\dots,\hat{\tau}_{D-1 l})^\T\), which consists of the maximum likelihood estimates for each \(\tau_d\). Then, each \(\hat{\tau}_{dl}\) is standardized to give \(D-1\) test statistics $Z_{dl} = \hat{\tau}_{dl}I_{dl}^{1/2}$, \(d=1,\dots,D-1\), with \(I_{dl} = \{\text{var}(\hat{\tau}_{dl})\}^{-1}\) the information level for treatment \(d\) at interim analysis \(l\). Since \(\hat{\boldsymbol{\tau}}_l\) is estimated via a normal linear model we know that $E(Z_{dl}) = \tau_dI_{dl}^{1/2}$ (Jennison and Turnbull, 2000).

Given fixed futility boundaries, \(f_{dl}\), and efficacy bounds, \(e_{dl}\), the following stopping rules are used at each analysis $l=1,...,L$, for each experimental treatment \(d=1,\dots,D-1\) satisfying \(f_{dm} \le Z_{dm} < e_{dm}\) for $m=1,\dots,l-1$
\begin{itemize}
	\item if \(Z_{dl}<f_{dl}\) treatment \(d\) is dropped without rejecting \(H_{0d}\),
	\item if \(f_{dl} \le Z_{dl}<e_{dl}\) the trial is continued with treatment \(d\) still present,
	\item and if \(e_{dl} \le Z_{dl}\) treatment \(d\) is dropped and \(H_{0d}\) rejected.
\end{itemize}
The control treatment, \(d=0\), remains present at every undertaken stage, and we only proceed to an additional stage if there is at least one experimental treatment remaining in the trial. It is convenient to take \(f_{dl}=f_l\) and \(e_{dl}=e_l\) for all \(d\) and \(l\), as well as \(f_L=e_L\) in order to ensure the trial conforms to the desired maximum number of stages and so that a conclusion is made for each $H_{0d}$. Note that rejection of one treatment's null hypothesis does not end the trial. Furthermore, with this formulation, once a treatment is dropped from the trial its standardised treatment effect is not tested at any future analyses.

In what follows, we will make use of the vectors \(\boldsymbol{\omega}_R = (\omega_{R1},\dots,\omega_{R D-1})^\T\) and \(\boldsymbol{\psi}_R = (\psi_{R1},\dots,\psi_{R D-1})^\T\). Here, \(\omega_{Rd} \in \{1,\dots,L\}\) is the analysis at which experimental treatment \(d\) was dropped from the trial. Moreover, \(\psi_{Rd} \in \{0,1\}\); with \(\psi_{Rd}=1\) if experimental treatment \(d\) was dropped for efficacy, and is 0 otherwise. Prior to a trials commencement $\boldsymbol{\omega}_R$ and $\boldsymbol{\psi}_R$ are unknown random variables. However, the probability that the trial progresses according to some particular $\boldsymbol{\omega}=(\omega_1,\dots,\omega_{D-1})^\T$ and $\boldsymbol{\psi}=(\psi_1,\dots,\psi_{D-1})^\T$, given a vector of true response rates \(\boldsymbol{\tau} = (\tau_1,\dots,\tau_{D-1})^\T\), can be computed using multivariate normal integration. More specifically, given this particular \((\boldsymbol{\omega},\boldsymbol{\psi})\) pair the covariance between, and the information level of, the test statistics can be computed and the following integral evaluated (see Jennison and Turnbull (2000) or Wason (2015) for further details)

\begin{small}
	\begin{equation*}
	\begin{split}
	\text{pr}(\boldsymbol{\omega}_R=\boldsymbol{\omega},\boldsymbol{\psi}_R=\boldsymbol{\psi}\mid\boldsymbol{\tau}) &= \int_{\text{l}(1,\omega_1,\psi_1)}^{\text{u}(1,\omega_1,\psi_1)} \dots\\
	& \qquad \dots \int_{\text{l}(L,\omega_{D-1},\psi_{D-1})}^{\text{u}(L,\omega_{D-1},\psi_{D-1})} \! \phi\left\{\boldsymbol{x}, \boldsymbol{r}(\boldsymbol{\tau},L) \circ \boldsymbol{I}^{1/2}_{(\boldsymbol{\omega},\boldsymbol{\psi})}, \boldsymbol{\Lambda}_{(\boldsymbol{\omega},\boldsymbol{\psi})}\right\} \,\\
	& \qquad \qquad \qquad \qquad \qquad \qquad 
	\qquad \mathrm{d}x_{L(D-1)}\dots\mathrm{d}x_{11},
	\end{split}
	\end{equation*}
\end{small}
where
\begin{itemize}
	\item  $\boldsymbol{x}=(x_{11},\dots,x_{1(D-1)},\dots,x_{L1},\dots,x_{L(D-1)})^\T$,
	\item \(\phi\{\boldsymbol{x},\boldsymbol{\mu},\boldsymbol{\Lambda}\}\) is the probability density function of a multivariate normal distribution with mean \(\boldsymbol{\mu}\) and covariance matrix \(\boldsymbol{\Lambda}\), evaluated at vector $\boldsymbol{x}$,
	\item \(\boldsymbol{r}(\boldsymbol{\tau},L)\) is the vector formed by repeating \(\boldsymbol{\tau}\) \(L\) times,
	\item $\boldsymbol{I}_{(\boldsymbol{\omega},\boldsymbol{\psi})} = \left(\boldsymbol{I}_{1,(\boldsymbol{\omega},\boldsymbol{\psi})}^\T, \dots, \boldsymbol{I}_{L,(\boldsymbol{\omega},\boldsymbol{\psi})}^\T\right)^\T$, where \(\boldsymbol{I}_{l,(\boldsymbol{\omega},\boldsymbol{\psi})} = (I_{1l},\dots,I_{(D-1) l})_{(\boldsymbol{\omega},\boldsymbol{\psi})}^\T\) is the vector of information levels for the estimated treatment effects at interim analysis \(l\), according to (conditional on) the particular \((\boldsymbol{\omega},\boldsymbol{\psi})\) being considered,
	\item \(\circ\) denotes the Hadamard product of two vectors,
	\item the square root of the vector $\boldsymbol{I}_{(\boldsymbol{\omega},\boldsymbol{\psi})}$ is taken in an element wise manner,
	\item $\text{l}$ and $\text{u}$ are functions that tell us the lower and upper integration limits for the test statistic $Z_{dl}$ given values for $l$, $\omega_d$ and $\psi_d$. For example, $\text{l}(1,2,1)=f_1$ and $\text{u}(1,2,1)=e_1$, whilst $\text{l}(2,2,1)=e_2$ and $\text{u}(2,2,1)=\infty$, and then $\text{l}(l,2,1)=-\infty$ and $\text{u}(l,2,1)=\infty$ for $l>2$,
	\item \(\boldsymbol{\Lambda}_{(\boldsymbol{\omega},\boldsymbol{\psi})}\) is the covariance matrix between the standardized test statistics at and across each interim analysis according to \((\boldsymbol{\omega},\boldsymbol{\psi})\). Thus, using \(\boldsymbol{Z}_l=(Z_{1l},\dots,Z_{D-1 l})^\T\), we have
	\[ \boldsymbol{\Lambda}_{(\boldsymbol{\omega},\boldsymbol{\psi})} =
	\begin{pmatrix}
	\text{cov}\left(\boldsymbol{Z}_1,\boldsymbol{Z}_1 \mid \boldsymbol{\omega},\boldsymbol{\psi}\right) & \dots & \text{cov}\left(\boldsymbol{Z}_1,\boldsymbol{Z}_L \mid \boldsymbol{\omega},\boldsymbol{\psi}\right) \\[0.3em]
	\vdots & \ddots & \vdots \\[0.3em]
	\text{cov}\left(\boldsymbol{Z}_L,\boldsymbol{Z}_1 \mid \boldsymbol{\omega},\boldsymbol{\psi}\right) & \dots & \text{cov}\left(\boldsymbol{Z}_L,\boldsymbol{Z}_L \mid \boldsymbol{\omega},\boldsymbol{\psi}\right)
	\end{pmatrix}. \]
	However, \(\boldsymbol{Z}_l=\hat{\boldsymbol{\tau}}_l \circ \boldsymbol{I}_{l,(\boldsymbol{\omega},\boldsymbol{\psi})}^{1/2} \), and by the properties of normal linear models \(\text{cov}(\hat{\boldsymbol{\tau}}_{l_1},\hat{\boldsymbol{\tau}}_{l_2} \mid \boldsymbol{\omega},\boldsymbol{\psi}) = \text{cov}(\hat{\boldsymbol{\tau}}_{l_2},\hat{\boldsymbol{\tau}}_{l_2} \mid \boldsymbol{\omega},\boldsymbol{\psi})\) \((l_1,l_2=1,\dots,L;\ l_1 \le l_2)\) (Jennison and Turnbull, 2000), giving
	\begin{footnotesize}
		\begin{equation} \label{eq:eq1}
		\text{cov}(\boldsymbol{Z}_{l_1},\boldsymbol{Z}_{l_2} \mid \boldsymbol{\omega},\boldsymbol{\psi}) = \text{diag}(\boldsymbol{I}_{l_1,(\boldsymbol{\omega},\boldsymbol{\psi})}^{1/2}) \text{cov}(\hat{\boldsymbol{\tau}}_{l_2},\hat{\boldsymbol{\tau}}_{l_2} \mid \boldsymbol{\omega},\boldsymbol{\psi}) \text{diag}(\boldsymbol{I}_{l_2,(\boldsymbol{\omega},\boldsymbol{\psi})}^{1/2}),
		\end{equation}
	\end{footnotesize}
	for \(l_1,l_2=1,\dots,L,\ l_1 \le l_2\), and where \(\text{diag}(\boldsymbol{v})\) is the matrix formed by placing the elements of vector \(\boldsymbol{v}\) along the leading diagonal.
\end{itemize}
Note that Equation~(\ref{eq:eq1}) in conjunction with the expectations of our standardised test statistics, and the observation that $(\boldsymbol{Z}_1^\T,\dots,\boldsymbol{Z}_L^\T)^\T$ is multivariate normal, can be restated simply as that our test statistics follow the canonical joint distribution (Jennison and Turnbull, 2000).

\subsection{Familywise Error Rate Control}

It is a common requirement of clinical trial designs that the probability of one or more false rejections within the family of null hypotheses is not greater than some \(\alpha\). This is known as strong control of the familywise error rate. In this section we establish strong control for our considered trial design.

To evaluate the familywise error rate of a design, for any $\boldsymbol{\tau}$, the above integral can be evaluated for all $\boldsymbol{\omega}$ and $\boldsymbol{\psi}$ that would imply a type-I error is made, and the results summed. In order to demonstrate how to strongly control though, it is essential to know the forms of the \(\boldsymbol{I}_{l,(\boldsymbol{\omega},\boldsymbol{\psi})}\) and \(\boldsymbol{\Lambda}_{(\boldsymbol{\omega},\boldsymbol{\psi})}\) for each \((\boldsymbol{\omega},\boldsymbol{\psi})\). However, by Equation~(\ref{eq:eq1}), the \(\boldsymbol{I}_{l,(\boldsymbol{\omega},\boldsymbol{\psi})}\) and \(\boldsymbol{\Lambda}_{(\boldsymbol{\omega},\boldsymbol{\psi})}\) can be determined if $\text{cov}(\hat{\boldsymbol{\beta}}_l,\hat{\boldsymbol{\beta}}_l \mid \boldsymbol{\omega},\boldsymbol{\psi})$ is known for $l=1,\dots,L$.

Thus, consider the matrix $\text{cov}(\hat{\boldsymbol{\beta}}_l,\hat{\boldsymbol{\beta}}_l \mid \boldsymbol{\omega},\boldsymbol{\psi})$ for some $l \le L$ and any $(\boldsymbol{\omega},\boldsymbol{\psi})$. We compute values for \(L_{lr}\) \((r=1,\dots,D)\); the number of stages of the trial, up to analysis $l$, in which \(r\) treatments were remaining. Since we do not continue the trial unless at least one experimental treatment remains, $L_{l1}=0$ always. It will be convenient however to still include this value. Moreover, it is clear that the $L_{lr}$ are uniquely determined given $(\boldsymbol{\omega},\boldsymbol{\psi})$. Now, $\text{cov}(\hat{\boldsymbol{\beta}}_l,\hat{\boldsymbol{\beta}}_l \mid \boldsymbol{\omega},\boldsymbol{\psi})$ can always be decomposed to be a sum over the determined \(L_{lr}\) and the pre-specified sequences \(S_r\) (see Fitzmaurice et al. (2011) for details)
\begin{align*}
\text{cov}(\hat{\boldsymbol{\beta}}_l,\hat{\boldsymbol{\beta}}_l \mid \boldsymbol{\omega},\boldsymbol{\psi}) &= \text{cov}(\hat{\boldsymbol{\beta}}_l,\hat{\boldsymbol{\beta}}_l \mid L_{l1},\dots,L_{lD}),\\
&= \left( \sum_{r=1}^{D} L_{lr} \frac{n}{\left|S_r\right|} \sum_{i=1}^{|S_r|} \boldsymbol{X}^\T_{s_{ri}} \boldsymbol{\Sigma}_r^{-1} \boldsymbol{X}_{s_{ri}} \right)^{-1}.
\end{align*}
Here \(\boldsymbol{X}_{s_{ri}}\) is the uniquely defined \( r \times (2D-1) \) design matrix for a single patient allocated to sequence \(s_{ri}\), and \(\boldsymbol{\Sigma}_r\) is the easily computed \(r \times r\) covariance matrix of the responses for a single patient allocated \(r\) treatments in total. The factor \(n/\left|S_r\right|\) arises from the number of patients allocated to each sequence \(s_{ri}\) by our choice of period balance.

We now establish two key results about $\text{cov}(\hat{\boldsymbol{\beta}}_l,\hat{\boldsymbol{\beta}}_l\mid L_{l1},\dots,L_{lD})$. Following this, we provide a proof detailing how to strongly control the familywise error rate.

\begin{thm}
	\label{thm1}
	Let \( \boldsymbol{\beta} = (\mu_0,\pi_2,\dots,\pi_D,\tau_1,\dots,\tau_{D-1})^\T \). Consider an analysis to be performed after some number of stages $l$. Then 
	\begin{enumerate}
		\item We have
		\begin{footnotesize}
			\begin{align}
			\text{cov}(\hat{\boldsymbol{\beta}}_l,\hat{\boldsymbol{\beta}}_l \mid L_{l1},\dots,L_{lD-1}=0,L_{lD}=l) &= \left( \frac{ln}{\left|S_D\right|} \sum_{i=1}^{\left|S_D\right|} \boldsymbol{X}_{s_{Di}}^\T \boldsymbol{\Sigma}_D^{-1} \boldsymbol{X}_{s_{Di}} \right)^{-1},\nonumber\\
			&= \frac{1}{ln} \begin{pmatrix} F & \boldsymbol{G}^\T & \boldsymbol{G}^\T \\[0.3em]
			\boldsymbol{G} & \boldsymbol{H} & \boldsymbol{0}_{D-1,D-1} \\[0.3em]
			\boldsymbol{G} & \boldsymbol{0}_{D-1,D-1} & \boldsymbol{H}
			\end{pmatrix},\label{eq:eq2}
			\end{align}
		\end{footnotesize}
		where
		\begin{align*}
		F &= \sigma_b^2 + \frac{2D-1}{D}\sigma_e^2,\\
		\boldsymbol{G}_{pq} &= -\sigma_e^2 & (p&=1,\dots,D-1;\ q=1), \\
		\boldsymbol{H}_{pq} &= \sigma_e^2(1 + \delta_{pq}) & (p&=1,\dots,D-1;\ q=1,\dots,D-1).
		\end{align*}
		\item If $q \ge 2$ is the largest integer such that $L_{lr} = 0$ for $r=1,\dots,q-1$, then the covariance of the estimates of the fixed effects \(\hat{\pi}_{2l},\dots,\hat{\pi}_{ql}, \hat{\tau}_{1l},\dots,\hat{\tau}_{q-1 l}\) is identical to that it would be for \(L_{l1}=\dots=L_{lD-1}=0\). Moreover, the covariance between the estimates of \(\hat{\pi}_{2l},\dots,\hat{\pi}_{ql}, \hat{\tau}_{1l},\dots,\hat{\tau}_{q-1 l}\) and the estimates of \(\hat{\pi}_{q+1l},\dots,\hat{\pi}_{Dl}, \hat{\tau}_{ql},\dots,\hat{\tau}_{D-1 l}\) is also identical to that it would be for \(L_{l1}=\dots=L_{lD-1}=0\).
	\end{enumerate}
\end{thm}
\begin{proof}
	See Appendix C.
\end{proof}

Note that part (1) of the above theorem implies
\[ cov(\hat{\tau}_{d_1l},\hat{\tau}_{d_2l} \mid L_{l1}=\dots=L_{lD-1}=0,L_{lD}=l) = \frac{\sigma_e^2}{ln}(1 + \delta_{d_1d_2}), \]
for \(d_1,d_2 \in \{1,\dots,D-1\}\). This is the familiar result for complete block sequences that there is no dependence upon the between patient variance \(\sigma_b^2\) (Jones and Kenward, 2014).

\begin{thm}
	\label{theorem2}
	A group sequential crossover trial of the type considered, with $D \ge 2$, testing the \(D-1\) hypotheses \(H_{0d} : \tau_d \le 0\), \(H_{1d} : \tau_d > 0\), attains a maximal value of its familywise error rate for \(\tau_1=\dots=\tau_{D-1}=0\).
\end{thm}
\begin{proof}
	Theorem~\ref{thm1} implies that elements of the covariance matrix \(\text{cov}(\hat{\boldsymbol{\tau}}_l,\hat{\boldsymbol{\tau}}_l)\) that differ from the case where no treatments have been dropped are exactly those corresponding to unstandardized test statistics no longer of importance. Consequently, the values of \( \boldsymbol{I}_{l,(\boldsymbol{\omega},\boldsymbol{\psi})} \) and \(\boldsymbol{\Lambda}_{(\boldsymbol{\omega},\boldsymbol{\psi})}\) that differ from the case \(\boldsymbol{\omega}=(L,\dots,L)^\T\) are only ever those corresponding to limits of integration given by \((-\infty,\infty)\) in our computation of \(\text{pr}(\boldsymbol{\omega}_R=\boldsymbol{\omega},\boldsymbol{\psi}_R=\boldsymbol{\psi}\mid\boldsymbol{\tau})\). By the marginal distribution properties of the multivariate normal distribution, we therefore need only consider one matrix \( \boldsymbol{\Lambda}_{(\boldsymbol{\omega},\boldsymbol{\psi})} \), and one set of vectors \( \boldsymbol{I}_{l,(\boldsymbol{\omega},\boldsymbol{\psi})} \) \((l=1,\dots,L)\); exactly those given by the case \(\boldsymbol{\omega}=(L,\dots,L)^\T\). Denote these by \(\boldsymbol{\Lambda}\) and \( \boldsymbol{I}_l\), and set \(\boldsymbol{I}=(\boldsymbol{I}_1^\T,\dots,\boldsymbol{I}_L^\T)^\T\). For more information on this, see Appendix A. We now have
	
	\begin{small}
		\begin{equation*}
		\begin{split}
		\text{pr}(\boldsymbol{\omega}_R=\boldsymbol{\omega},\boldsymbol{\psi}_R=\boldsymbol{\psi}\mid\boldsymbol{\tau}) &= \int_{\text{l}(1,\omega_1,\psi_1)}^{\text{u}(1,\omega_1,\psi_1)} \dots\\
		& \qquad \dots \int_{\text{l}(L,\omega_{D-1},\psi_{D-1})}^{\text{u}(L,\omega_{D-1},\psi_{D-1})} \! \phi\left\{\boldsymbol{x}, \boldsymbol{r}(\boldsymbol{\tau},L) \circ \boldsymbol{I}^{1/2}, \boldsymbol{\Lambda}\right\} \,\\ & \qquad \qquad \qquad \qquad \qquad \qquad \qquad \mathrm{d}x_{L(D-1)}\dots\mathrm{d}x_{11}.
		\end{split}
		\end{equation*}
	\end{small}
	Now, consider without loss of generality the probability we reject \(H_{01}\), and denote by \(\Omega\) and \(\Psi\) the sets of all possible \(\boldsymbol{\omega}\) and \(\boldsymbol{\psi}\) respectively. By integrating over all possible values of \(\omega_2,\dots,\omega_{D-1}\) and \(\psi_2,\dots,\psi_{D-1}\), we have that the probability we reject each \(H_{01}\) does not depend on the values of $\tau_2,\dots,\tau_{D-1}$, i.e. on the other treatments tested
	
	\begin{small}
		\begin{align*}
		\text{pr}\left( \text{Reject } H_{01} \mid \boldsymbol{\tau} \right)&= \sum_{\{\boldsymbol{\psi} \in \Psi : \psi_1=1\}} \sum_{\boldsymbol{\omega} \in \Omega} \text{pr}(\boldsymbol{\omega}_R=\boldsymbol{\omega},\boldsymbol{\psi}_R=\boldsymbol{\psi}\mid\boldsymbol{\tau}), \displaybreak[0] \\
		&= \sum_{\{\boldsymbol{\psi} \in \Psi : \psi_1=1\}} \sum_{\boldsymbol{\omega} \in \Omega} \int_{\text{l}(1,\omega_1,\psi_1)}^{\text{u}(1,\omega_1,\psi_1)} \dots \\
		& \qquad \qquad \qquad \dots \int_{\text{l}(L,\omega_{D-1},\psi_{D-1})}^{\text{u}(L,\omega_{D-1},\psi_{D-1})} \! \phi\left\{\boldsymbol{x}, \boldsymbol{r}(\boldsymbol{\tau},L) \circ \boldsymbol{I}^{1/2}, \boldsymbol{\Lambda}\right\} \,\\ & \qquad \qquad \qquad \qquad \qquad \qquad \qquad \qquad \qquad \mathrm{d}x_{L(D-1)}\dots\mathrm{d}x_{11}, \displaybreak[0] \\
		&= \sum_{\omega_1=1}^L \int_{\text{l}(1,\omega_1,1)}^{\text{u}(1,\omega_1,1)} \dots \\
		& \qquad \qquad \dots \int_{\text{l}(L,\omega_1,1)}^{\text{u}(L,\omega_1,1)} \! \phi\left\{(x_{11},\dots,x_{L1})^\T, \boldsymbol{r}(\tau_1,L) \circ \boldsymbol{I}^{1/2}_{\tau_1}, \boldsymbol{\Lambda}_{\tau_1}\right\} \,\\ & \qquad \qquad \qquad \qquad \qquad \qquad \qquad \mathrm{d}x_{L1}\dots\mathrm{d}x_{11},
		\end{align*}
	\end{small}
	where \(\boldsymbol{I}_{\tau_1}\) and \(\boldsymbol{\Lambda}_{\tau_1}\) are the restrictions of \(\boldsymbol{I}\) and \(\boldsymbol{\Lambda}\) to rows and columns corresponding to experimental treatment \(d=1\) respectively. This final form for $\text{pr}\left( \text{Reject } H_{01} \mid \boldsymbol{\tau} \right)$ is identical to that it would be in the case $D=2$. Therefore to ascertain the $\boldsymbol{\tau}$ giving the maximal familywise error rate of a trial with $D \ge 2$, it suffices to consider which $\tau_* \le 0$ maximises the probability $H_{01}$ is rejected in a trial with $D=2$ initial treatments. For then, $\boldsymbol{\tau} = (\tau_*,\dots,\tau_*)^\T$ using this $\tau_*$, will provide the maximum probability of rejecting at least one true $H_{0d}$ for some $d$, i.e. the maximum familywise error rate. To see this, consider the familywise error rate for $\boldsymbol{\tau} = (\tau_*,\dots,\tau_*)^\T$. If one changes some individual element $\tau_{d_1}$ of this vector, this does not effect the probability that $H_{0d_2}$ is rejected for $d_2 \neq d_1$, and it can only decrease the probability that $H_{0d_1}$ is incorrectly rejected. Thus overall, straying from this $\boldsymbol{\tau} = (\tau_*,\dots,\tau_*)^\T$ can only decrease the familywise error rate.
	
	Thus, now consider all possible realisations of the test statistics of a trial with $D=2$, and their associated values of \((\boldsymbol{\omega},\boldsymbol{\psi})=(\omega_1,\psi_1)\). We have \(\boldsymbol{Z}=(Z_{11},\dots,Z_{1L})^\T \in \mathbb{R}^L\), with \(Z_{1L} = \dots = Z_{1 \omega_1} \) if the trial was stopped at stage \(\omega_1\). Now consider increasing the value of the test statistics by some \(\eta > 0\). All instances before where \(H_{01}\) was rejected will still exceed the efficacy bound of that stage, or earlier, and so \(H_{01}\) will still be rejected. Therefore, the probability of rejecting \(H_{01}\) is at least as large as before. Thus, increasing the value of \(\tau_1 \le 0 \) causes a non-decreasing change in the value of the type-I error rate. Therefore, the probability of rejecting $H_{01}$ is maximized by \(\tau_1=0\); implying in turn that the maximal familywise error rate of a trial with $D \ge 2$ is given by \(\boldsymbol{\tau}=(\tau_1,\dots,\tau_{D-1})^\T=(0,\dots,0)^\T\).
\end{proof}

\subsection{Design Characteristics}

A trial will now be fully specified given values for $D$, $L$, $\sigma_e^2$ and $n$, as well as choices for the $S_r$, and the futility and efficacy boundaries, \(f_1,\dots,f_L\) and \(e_1,\dots,e_L\) respectively. Given these, \(\boldsymbol{\Lambda}\) and \(\boldsymbol{I}\) can be computed using the results above. Then, by Theorem~\ref{theorem2} we can strongly control the familywise error rate to \(\alpha\) for this design using the following sum of integrals
\begin{align*}
\alpha &= \sum_{\{\boldsymbol{\psi} \in \Psi : \Sigma_d\psi_d>0\}}\sum_{\boldsymbol{\omega} \in \Omega}\int_{\text{l}(1,\omega_1,\psi_1)}^{\text{u}(1,\omega_1,\psi_1)} \dots \\
& \qquad \qquad \dots \int_{\text{l}(L,\omega_{D-1},\psi_{D-1})}^{\text{u}(L,\omega_{D-1},\psi_{D-1})} \! \phi\left\{\boldsymbol{x}, \boldsymbol{r}(0,L(D-1)), \boldsymbol{\Lambda}\right\} \, \mathrm{d}x_{L(D-1)}\dots\mathrm{d}x_{11},
\end{align*}
Additionally, suppose that we wish to power this trial to reject a particular null hypothesis, without loss of generality \(H_{01}\), at some clinically relevant difference \(\tau_1=\delta\). The type-II error rate \(\beta\) for $H_{11}$ is then given by
\begin{align*}
\beta &= 1 - \sum_{\omega_1=1}^L \int_{\text{l}(1,\omega_1,1)}^{\text{u}(1,\omega_1,1)} \dots \\ & \qquad \qquad \qquad \dots \int_{\text{l}(L,\omega_1,1)}^{\text{u}(L,\omega_1,1)} \! \phi\left\{(x_{11},\dots,x_{L1})^\T, \boldsymbol{r}(\delta,L) \circ \boldsymbol{I}^{1/2}_{\tau_1}, \boldsymbol{\Lambda}_{\tau_1}\right\} \,\\ & \qquad \qquad \qquad \qquad \qquad \qquad \qquad \qquad \mathrm{d}x_{L1}\dots\mathrm{d}x_{11}.
\end{align*}
Moreover, denoting by $N$ and $O$ the total number of patients and observations required by the trial respectively, we can compute the expected sample size, \(E(N\mid\boldsymbol{\tau})\), or expected number of observations, \(E(O\mid\boldsymbol{\tau})\), for any \(\boldsymbol{\tau}\), according to
\begin{align*}
E(N\mid\boldsymbol{\tau}) &= \sum_{\boldsymbol{\psi}\in\Psi} \sum_{\boldsymbol{\omega}\in\Omega} \text{pr}(\boldsymbol{\omega}_R=\boldsymbol{\omega},\boldsymbol{\psi}_R=\boldsymbol{\psi}\mid\boldsymbol{\tau})\text{N}(\boldsymbol{\omega},\boldsymbol{\psi}),\\ E(O\mid\boldsymbol{\tau}) &= \sum_{\boldsymbol{\psi}\in\Psi} \sum_{\boldsymbol{\omega}\in\Omega} \text{pr}(\boldsymbol{\omega}_R=\boldsymbol{\omega},\boldsymbol{\psi}_R=\boldsymbol{\psi}\mid\boldsymbol{\tau})\text{O}(\boldsymbol{\omega},\boldsymbol{\psi}).
\end{align*}
Here, \(\text{N}(\boldsymbol{\omega},\boldsymbol{\psi})\) and \(\text{O}(\boldsymbol{\omega},\boldsymbol{\psi})\) are functions that give the number of patients and observations respectively, required by a trial that progresses according to \((\boldsymbol{\omega},\boldsymbol{\psi})\). Specifically
\begin{align*}
\text{N}(\boldsymbol{\omega},\boldsymbol{\psi}) &= n \max_{\{d=1,\dots,D-1\}} \omega_d,\\ \text{O}(\boldsymbol{\omega},\boldsymbol{\psi}) &= n \sum_{l=1}^L \left(\sum_{d=1}^{D-1}\mathbb{I}_{\{\omega_d \ge l\}} + 1\right),
\end{align*}
where $\mathbb{I}_{\{\omega_d \ge l\}} = 1$ if $\omega_d \ge l$, and is 0 otherwise.

\section{EXAMPLE: TOMADO}
\label{s:example}

As an example of how to design a group sequential crossover trial with strong control of the familywise error rate, we will make use of the TOMADO crossover randomized controlled trial (Quinnell et al., 2014). This open-label trial compared three experimental treatments to a single control for the treatment of sleep apnoea-hypopnoea using a four-treatment four-period crossover design. The normally distributed secondary endpoint, Epworth Sleepiness Scale, hoped to observe negative test statistics. Therefore, we consider the decrease as the endpoint in order to retain the same hypothesis tests as before $H_{0d} : \tau_d \le 0, H_{1d} : \tau_d > 0,$ \(d=1,2,3\). The trial planned to recruit 90 patients, and utilising restricted error maximum likelihood estimation, the final analysis calculated that \(\sigma_e^2=6.51\). Taking this variance as the truth, the trial had a familywise error rate \(\alpha=0.05\) for \(\boldsymbol{\tau}=(\tau_1,\tau_2,\tau_3)^\T=(0,0,0)^\T\), and \(\beta=0.2\) for \(H_{11}\) at \(\tau_1=1.11\).

Many methods exist for determining boundaries for a one-sided group sequential trial with parallel treatment arms. Here, we consider analogues of the power family boundaries of Pampallona and Tsiatis (1994). For this, values for the desired type-I and type-II error rates, a clinically relevant difference \(\delta\), the maximum number of stages \(L\), the within person variance $\sigma_e^2$, and a shape parameter \(\Delta\) must be specified. A 2-dimensional grid search is then used to find the exact required maximal sample size. From this a suitable value of \(n\) is identified by rounding up to the nearest integer such that \(n\) is as required divisible by \(|S_2|,\dots,|S_D|\). Utilizing Williams squares for our designs, \(n\) was forced to be divisible by 12.

Taking \(\alpha=0.05\), \(\beta=0.2\), \(\delta=1.11\), $\sigma_e^2=6.51$, \(L=3\), and \(\Delta=-0.25\), 0, 0.5, 0.5 as examples, group sequential crossover trial designs were determined and compared to the single-stage design used by TOMADO. All computations were done in R (R Core Team, 2016) using the package groupSeqCrossover, available from the corresponding author upon request. Matlab (The Mathworks Inc., 2016) code employing symbolic algebra is also available to return the matrices given by several of the equations in the text. Use of both the R and Matlab code is detailed in Appendix D.

A summary of the performance of the designs is provided in Table~\ref{tab1}, and their computed boundaries are displayed in Figure~\ref{fig1}. We can see that, as is the case for two-arm parallel trial designs, there is a trend that larger values of \(\Delta\) result in larger maximum sample sizes and lower expected sample sizes due to their larger stopping regions. However, this is not the case for \(\Delta=0.25\) because of the requirement to round to a suitable integer value of \(n\).

Plots of the probability of rejecting \(H_{01}\), and rejecting \(H_{0d}\) for some \(d=1,2,3\), are provided for a range of values of \(\theta\) when \(\boldsymbol{\tau}=(\theta,\theta,\theta)^\T\) in Figure~\ref{fig2}. The power curves are similar for all the designs, with the only differences a result of rounding in the group sequential designs to achieve suitable values of \(n\).

As is to be expected for group sequential designs, the maximum sample size and maximum number of observations is larger than for the single-stage design. However, the group sequential designs have lower expected sample sizes under the global null hypothesis \((\boldsymbol{\tau}=\boldsymbol{0}=(0,0,0)^\T)\); up to a maximum of 23\% for \(\Delta=0.5\). Though, this comes at the expense of an increased expected sample size under the global alternative hypothesis \((\boldsymbol{\tau}=\boldsymbol{\delta}=(\delta,\delta,\delta)^\T)\).

From Figure~\ref{fig3}, the expected sample sizes of the group sequential designs can be seen to be far lower than the single-stage design for more extreme values of \(\theta\). A similar statement holds for the expected number of observations. However, in this instance for \(\Delta=0\), 0.5, the performance of the group sequential designs is better than the single-stage design across all values of \(\theta\).

\begin{table*}
	\caption{Example Design Performance. Summary of the performance of the single-stage and considered group sequential designs. The number of decimal places displayed in each row indicates the number to which rounding was performed}
	\label{tab1}
	\resizebox{\textwidth}{!}{
		\begin{tabular}{lrrrrr}
			\hline
			& \multicolumn{5}{c}{Design} \\
			& Single-stage & \(\Delta=-0.25\) & \(\Delta=0\) & \(\Delta=0.25\) & \(\Delta=0.5\) \\
			\hline
			\(n\) & 90 & 36 & 36 & 48 & 48 \\
			\(\text{pr}\left(\text{Reject } H_{01}\mid\boldsymbol{\tau}=\boldsymbol{0}\right)\) & 0.02 & 0.02 & 0.02 & 0.02 & 0.02 \\
			\(\text{pr}\left(\text{Reject } H_{01}\mid\boldsymbol{\tau}=\boldsymbol{\delta}\right)\) & 0.80 & 0.85 & 0.83 & 0.90 & 0.83 \\
			\(\text{pr}\left(\text{Reject } H_{0d} \text{ for some } d\mid\boldsymbol{\tau}=\boldsymbol{0}\right)\) & 0.05 & 0.05 & 0.05 & 0.05 & 0.05 \\
			\(\text{pr}\left(\text{Reject } H_{0d} \text{ for some } d\mid\boldsymbol{\tau}=\boldsymbol{\delta}\right)\) & 0.95 & 0.97 & 0.97 & 0.98 & 0.97 \\
			\(E\left(N\mid\boldsymbol{\tau}=\boldsymbol{0}\right)\) & 90.0 & 76.8 & 70.0 & 82.6 & 69.6 \\
			\(E\left(N\mid\boldsymbol{\tau}=\boldsymbol{\delta}\right)\) & 90.0 & 100.3 & 95.7 & 110.7 & 98.9 \\
			\(E\left(O\mid\boldsymbol{\tau}=\boldsymbol{0}\right)\) & 360.0 & 269.3 & 240.3 & 283.1 & 244.5 \\
			\(E\left(O\mid\boldsymbol{\tau}=\boldsymbol{\delta}\right)\) & 360.0 & 367.2 & 341.8 & 380.4 & 327.7 \\
			\(\max N\) & 90 & 108 & 108 & 144 & 144 \\
			\(\max O\) & 360 & 432 & 432 & 576 & 576 \\
			\hline
		\end{tabular}}
	\end{table*}
	
	\begin{figure}
		\vspace{6pc}
		\includegraphics[width=\textwidth]{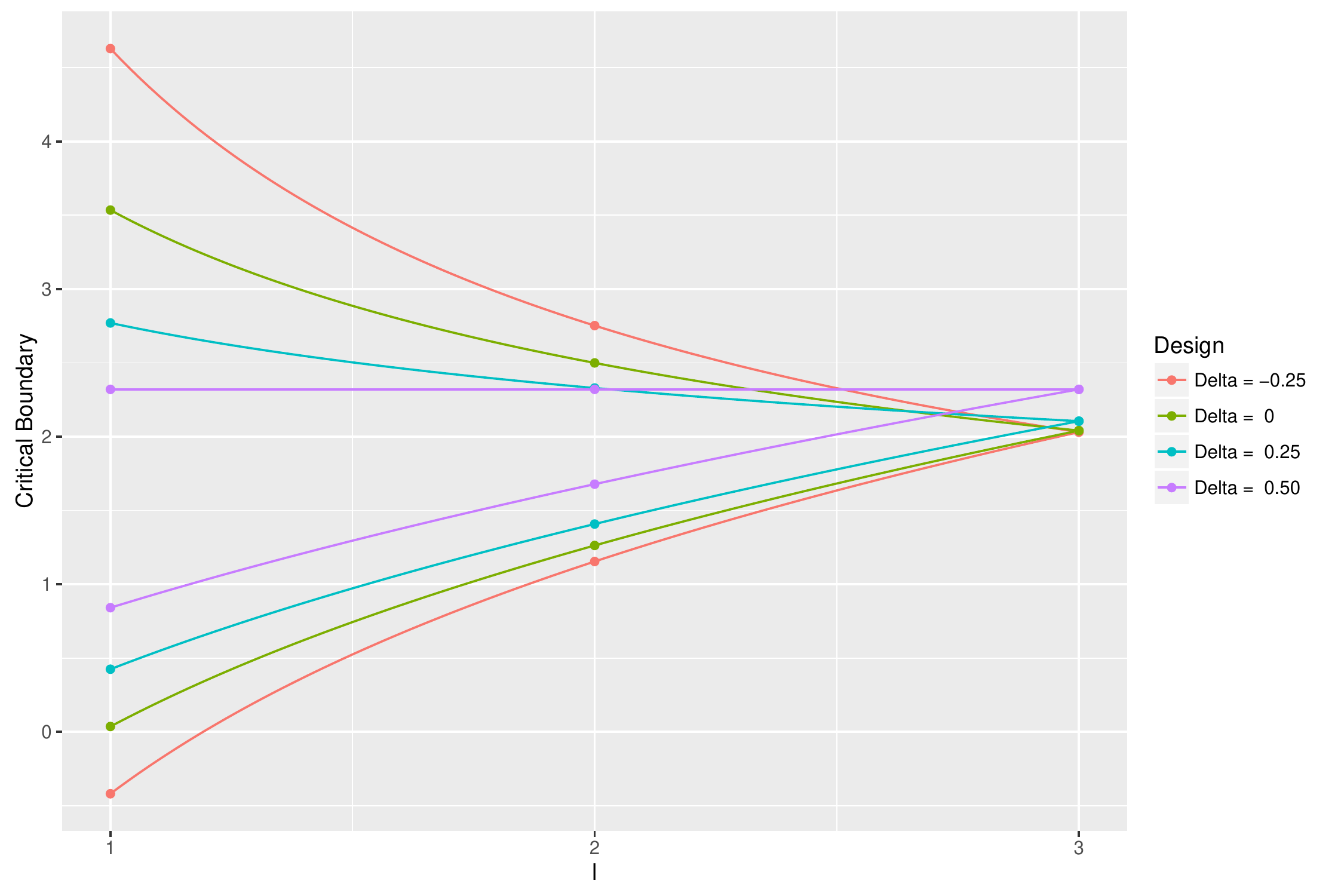}
		\caption{Stopping Boundaries. Computed efficacy and futility boundaries of the considered group sequential designs.}
		\label{fig1}
	\end{figure}
	
	\begin{figure}
		\vspace{6pc}
		\subfloat{\includegraphics[width=0.9\textwidth]{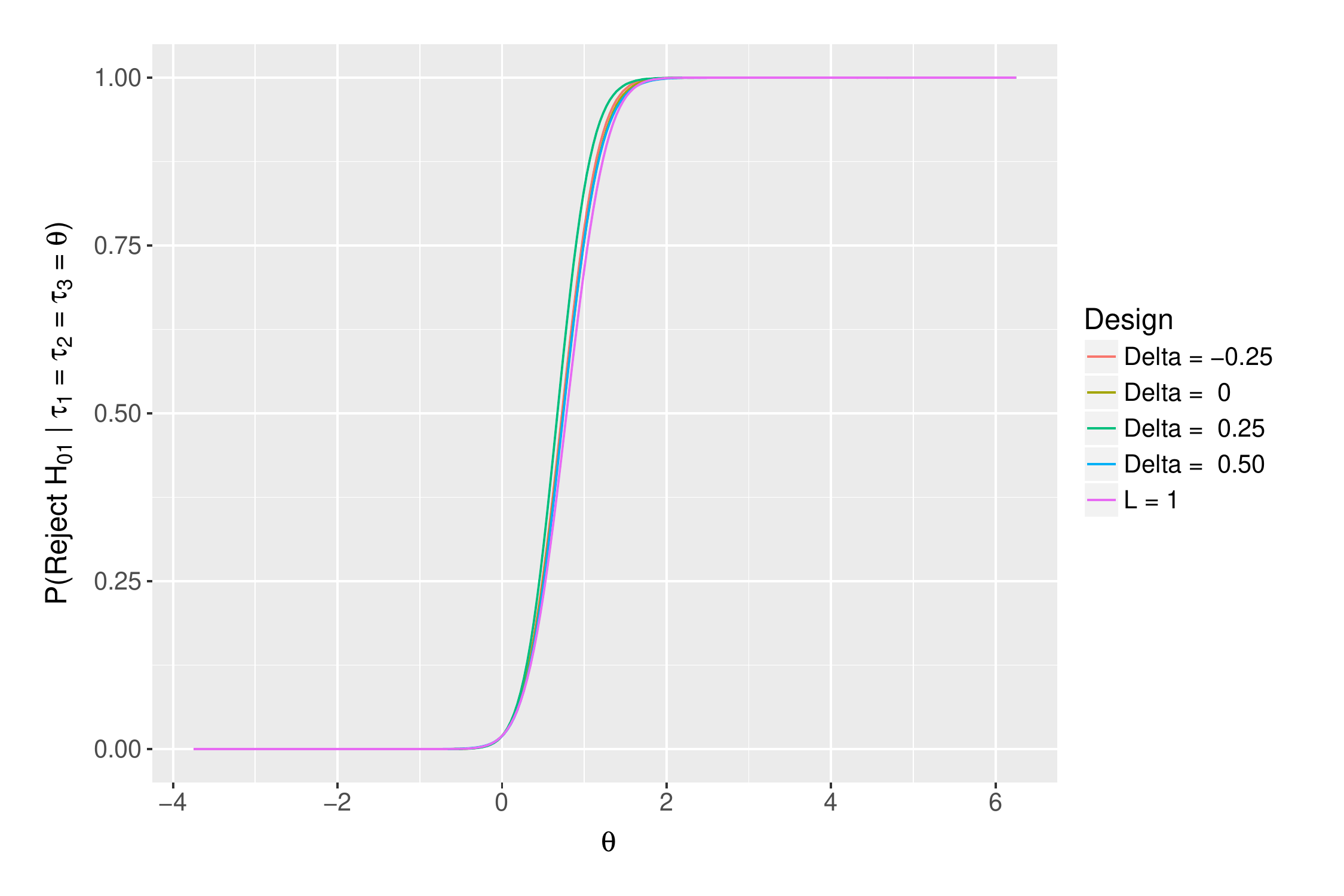}}
		\qquad
		\qquad
		\subfloat{\includegraphics[width=0.9\textwidth]{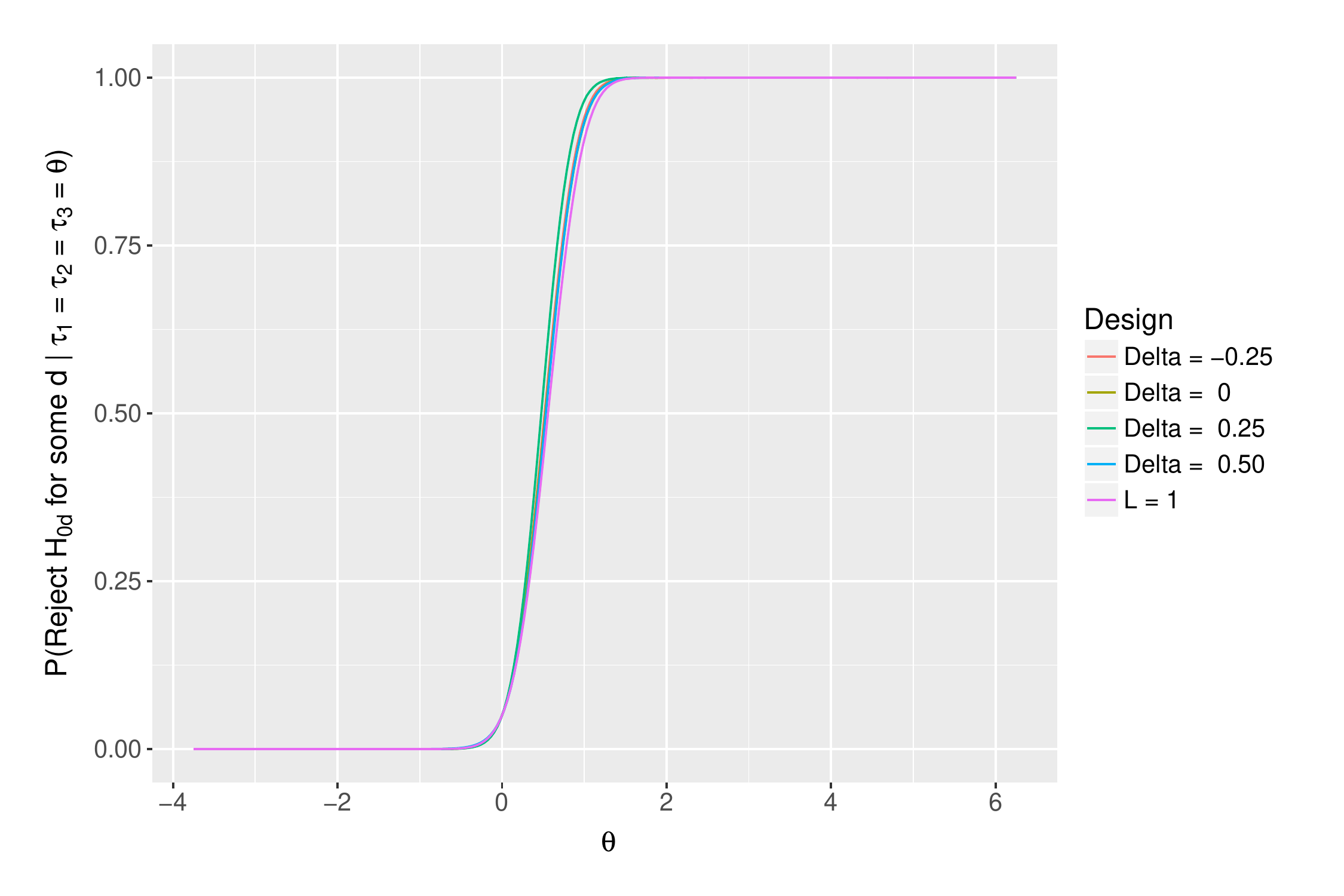}}
		\caption{Power Curves. Power curves of the single-stage ($L=1$) and considered group sequential designs across a range of values of the true response rate in the experimental treatment arms \(\theta\).}
		\label{fig2}
	\end{figure}
	
	\begin{figure}
		\vspace{6pc}
		\subfloat{\includegraphics[width=0.9\textwidth]{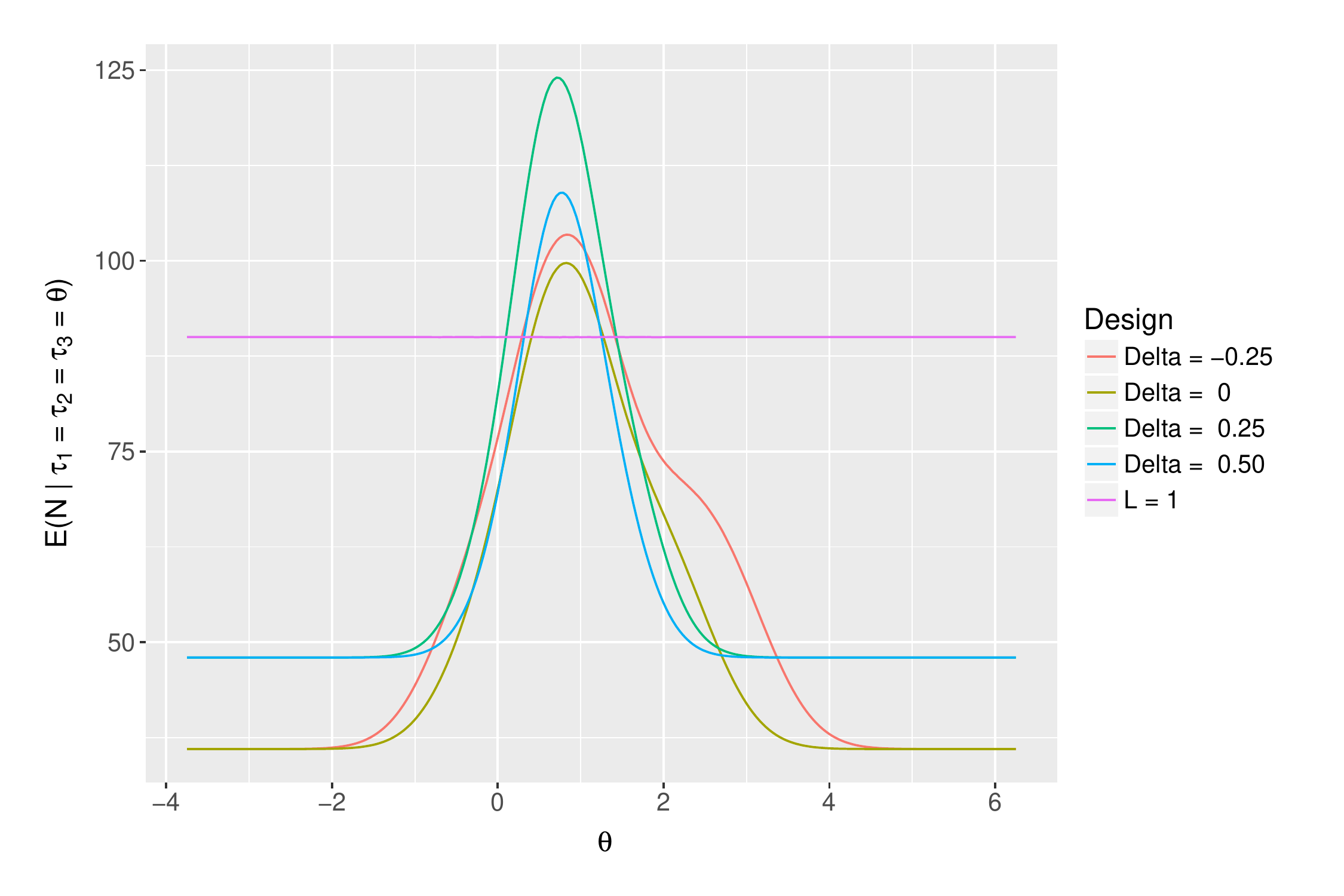}}
		\qquad
		\qquad
		\subfloat{\includegraphics[width=0.9\textwidth]{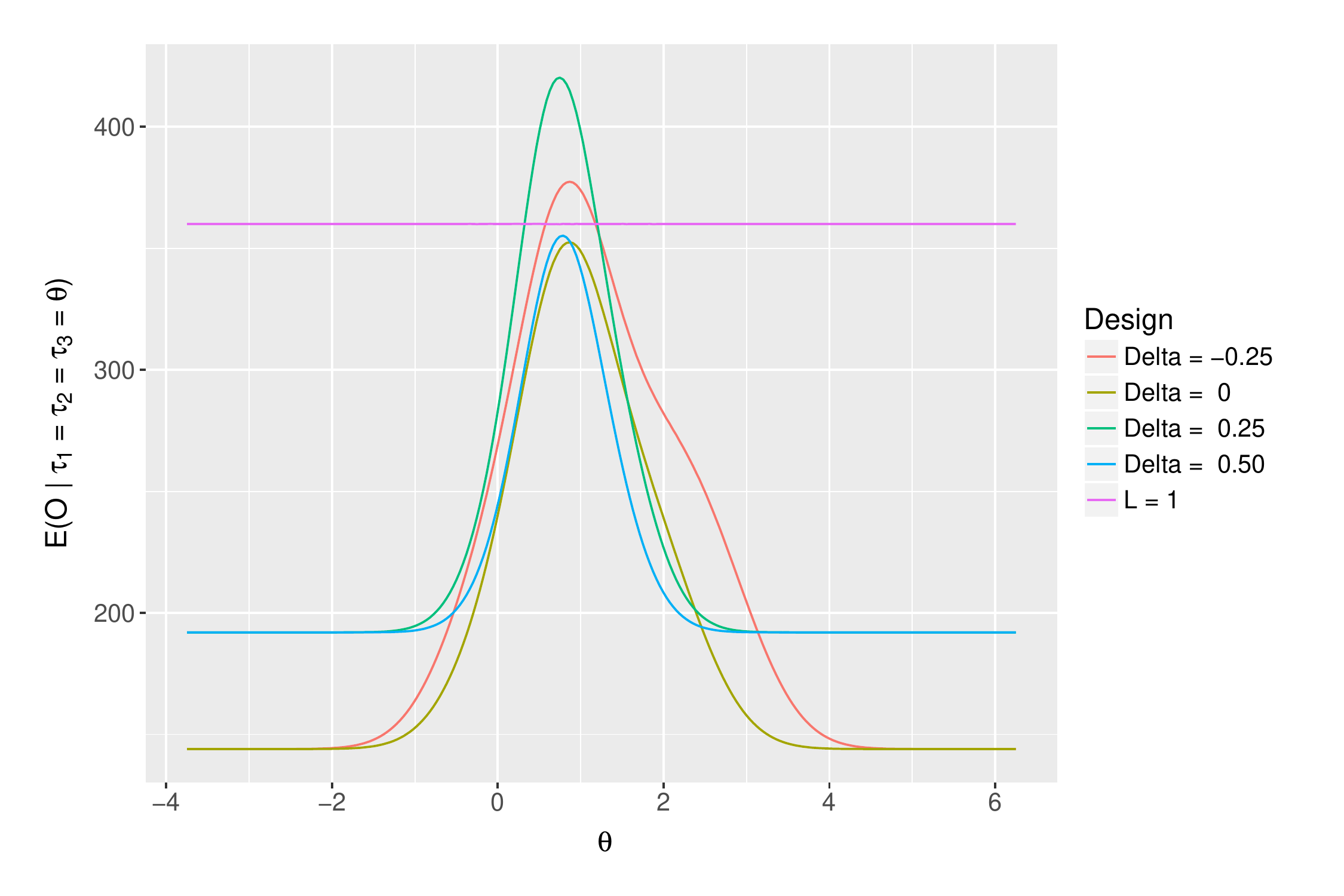}}
		\caption{Performance Measurement Curves. Expected sample size and expected number of observations curves of the single-stage ($L=1$) and considered group sequential designs across a range of values of the true response rate in the experimental treatment arms \(\theta\).}
		\label{fig3}
	\end{figure}
	
	\section{DISCUSSION}
	\label{s:discussion}
	
	There is a long history on group sequential clinical trials. Very few however utilise a crossover design. This may at least in part be due to no formal proof existing for how to strongly control the familywise error rate of such a trial. Here, we provided such a proof and then explored the performance of several sequential designs for the TOMADO trial.
	
	The expected sample size of the sequential designs was observed to be far lower than that of the single-stage design for a large range of values of the true response rate on all experimental treatments. Unfortunately, but unsurprisingly given the trial is not stopped unless all experimental treatments are dropped, there are regions in which the sequential designs are less efficient. Indeed, this region includes some values of $\theta$ between 0 and $\delta$, which may be more realistic observed treatment effects. However, for some considered designs this region is very small and does not include values near 0, which is notable for ethical reasons. This issue could even be further alleviated by utilising optimal stopping boundaries, as has been proposed for parallel arm designs (Wason and Jaki, 2012; Wason et al., 2012). Importantly, several of the designs always performed better than the single-stage design in terms of the expected number of observations required, which could be a significant factor in the cost and length of a trial. Consequently, we can conclude that a group sequential approach to a crossover trial improves efficiency in some circumstances.
	
	Several possible extensions to our work present themselves. For example, we assumed that period was reset in each trial stage. This could reflect a scenario where it is believed being enrolled in the trial will alter a patient's behaviour. However, in some cases, such as to deal with seasonal effects, it would be preferential to have different period effects in each stage.
	
	One simple extension would be to non-inferiority tests, from our present superiority testing framework. Non-inferiority tests, seeking to determine if new treatments are not clinically worse than an established control, would have hypotheses shifted by some factor from the ones presented here. Theorem~\ref{theorem2} could easily be altered to accommodate this, and then popular methods for boundary determination in this setting applied.
	
	Here, we have worked under an idealised scenario, assuming the within patient variance to be known prior to trial commencement. Though this is a common assumption in group sequential theory, it does bring limitations, since often a good estimate for the key variance parameter cannot be provided at the design stage. In this instance group sequential t-tests would almost certainly be required. Furthermore, simulation is required to quantify error rates accurately in the case of small sample sizes. To explore this scenario we analysed the true familywise error rate under the global null hypothesis of a particular design motivated again by the TOMADO trial, but with $L=2$ and $n=12$. We found that provided restricted error maximum likelihood was utilised, there was very little inflation in the familywise error rate over the nominal level $\alpha$. Details of this are provided in Appendix B.
	
	Moreover, we have only explored designing group sequential crossover trials. It is well known that if a final analysis is performed on data acquired in a sequential trial, not taking in to account the sequential nature, then biased treatment effects will be acquired. Extending established methodology for parameter estimation to our scenario will thus be important.
	
	Finally, we have implicitly assumed that there will be no patient drop out, and have not discussed the issue of patient recruitment rates. Though these are problems for all adaptive designs it is important to give them note. Owing to our need for one stages data to be analysed before the commencement of the following stage, it is likely the length of a trial using our approach would be longer for certain recruitment rates. It could be that recruitment is paused at interim, or that patients are continually recruited under the old scheme until results are available, which would lead to overrun and an increase in the expected number of observations and sample size. Thus this would be an important factor to consider when choosing an appropriate design for a trial.
	
	Nevertheless, for future crossover trials, consideration should be given to a group sequential approach. This may assist substantially in the efficient prioritisation of efficacious treatments.
	
	\section*{APPENDIX A: FURTHER TECHNICAL DETAILS}\label{appA}

As discussed in Section \ref{s:model}, part (1) of Theorem~\ref{thm1} implies that
\[ \text{cov}(\hat{\tau}_{d_1l},\hat{\tau}_{d_2l} \mid L_{l1},\dots,L_{lD-1}=0,L_{lD}=l) = \frac{\sigma_e^2}{ln}(1 + \delta_{d_1d_2}), \]
for \(d_1,d_2 \in \{1,\dots,D-1\}\). Alternatively, it tells us that in this case $I_{dl} = ln/(2\sigma_e^2)$, for $d=1,\dots,D-1$.

Moreover, using the above along with Equation~(\ref{eq:eq1}), in conjunction with part (2) of Theorem~\ref{thm1}, we have that if \(f_{p} \le Z_{d_1p} < e_{p}\) for $p=0,\dots,l_1-1$ (i.e. if treatment $d_1$ is present up to stage $l_1$) and \(f_{q} \le Z_{d_2q} < e_{q}\) for $q=0,\dots,l_2-1$ (i.e. if treatment $d_2$ is present up to stage $l_2$), with $l_1 \le l_2$, $l_1,l_2 \in \{1,\dots,L\}$ and $d_1,d_2 \in \{1,\dots,D-1\}$ (taking $f_0=-\infty$ and $e_0=\infty$) then
\begin{align*}
I_{d_1a} &= \frac{an}{2\sigma_e^2}, \displaybreak[0]\\
I_{d_2b} &= \frac{bn}{2\sigma_e^2}, \displaybreak[0]\\
\text{cov}(Z_{d_1a},Z_{d_2b} \mid \boldsymbol{\omega},\boldsymbol{\psi}) &= I_{d_1a}^{1/2}\text{cov}(\hat{\tau}_{d_1b},\hat{\tau}_{d_2b} \mid L_{b1},\dots,L_{bD})I_{d_2b}^{1/2}\nonumber\\ &= \left(\frac{an}{2\sigma_e^2}\right)^{1/2}\frac{\sigma_e^2}{bn}(1 + \delta_{d_1d_2})\left(\frac{bn}{2\sigma_e^2}\right)^{1/2},\nonumber\\
&= \frac{1}{2}\left(\frac{a}{b}\right)^{1/2}(1 + \delta_{d_1d_2}),
\end{align*}
for $a \le b$, $a = 1,\dots,l_1$, and $b = 1,\dots,l_2$.

For further clarity, as an example, consider the case $D=3$, $L=2$, and the associated value of $\text{pr}\left(\boldsymbol{\omega}_R=\boldsymbol{\omega},\boldsymbol{\psi}_R=\boldsymbol{\psi} \mid \boldsymbol{\tau} \right)$ when $\boldsymbol{\omega}=(2,1)^\T$ and $\boldsymbol{\psi}=(1,0)^\T$. Using the above we know the following elements of the matrix $\boldsymbol{\Lambda}_{(\boldsymbol{\omega},\boldsymbol{\psi})}$ and vector $\boldsymbol{I}_{(\boldsymbol{\omega},\boldsymbol{\psi})}=(\boldsymbol{I}_{1,(\boldsymbol{\omega},\boldsymbol{\psi})}^\T,\boldsymbol{I}_{2,(\boldsymbol{\omega},\boldsymbol{\psi})}^\T)^\T$
\begin{align*}
\boldsymbol{\Lambda}_{(\boldsymbol{\omega},\boldsymbol{\psi})} &= \begin{pmatrix}
1 &                               \frac{1}{2} &            \left(\frac{1}{2}\right)^{1/2} & \bullet \\[0.3em]
\frac{1}{2} &                                         1 & \frac{1}{2}\left(\frac{1}{2}\right)^{1/2} & \bullet \\[0.3em]
\left(\frac{1}{2}\right)^{1/2} & \frac{1}{2}\left(\frac{1}{2}\right)^{1/2} &                                         1 & \bullet \\[0.3em]
\bullet &                                         \bullet &                                         \bullet & \bullet \end{pmatrix},\\
\boldsymbol{I}_{(\boldsymbol{\omega},\boldsymbol{\psi})} &= \left(\frac{n}{2\sigma_e^2},\frac{n}{2\sigma_e^2},\frac{2n}{2\sigma_e^2},\bullet\right)^\T,
\end{align*}
where we have used $\bullet$ to signify an element we do not know the value of.

Now consider our computation of $\text{pr}\left(\boldsymbol{\omega}_R=\boldsymbol{\omega},\boldsymbol{\psi}_R=\boldsymbol{\psi} \mid \boldsymbol{\tau} \right)$. We have

\begin{small}
	\begin{align*}
	\text{pr}\left(\boldsymbol{\omega}_R=\boldsymbol{\omega},\boldsymbol{\psi}_R=\boldsymbol{\psi} \mid \boldsymbol{\tau} \right) = \int_{f_1}^{e_1}\int_{-\infty}^{f_1}\int_{e_2}^{\infty}\int_{-\infty}^{\infty} \ \phi\left\{\boldsymbol{x},\boldsymbol{r}(\boldsymbol{\tau},2)\circ \boldsymbol{I}_{(\boldsymbol{\omega},\boldsymbol{\psi})}, \boldsymbol{\Lambda}_{(\boldsymbol{\omega},\boldsymbol{\psi})} \right\} \mathrm{d}x_{22}\mathrm{d}x_{12}\mathrm{d}x_{21}\mathrm{d}x_{11}.
	\end{align*}
\end{small}
As we have seen we do not know the values of the final row and column of the matrix $\boldsymbol{\Lambda}_{(\boldsymbol{\omega},\boldsymbol{\psi})}$, or the final element of the vector $\boldsymbol{I}_{(\boldsymbol{\omega},\boldsymbol{\psi})}$. But, the fact mentioned in Theorem 2.2 becomes clear: this does not matter as the limits of integration corresponding to this variable are $(-\infty,\infty)$. Indeed, by the marginal distribution properties of the multivariate normal distribution, we need only as stated consider one matrix \( \boldsymbol{\Lambda}_{(\boldsymbol{\omega},\boldsymbol{\psi})} \), and one set of vectors \( \boldsymbol{I}_{l,(\boldsymbol{\omega},\boldsymbol{\psi})} \); exactly those given by the case \(\boldsymbol{\omega}=(L,\dots,L)^\T\). We denote these by \(\boldsymbol{\Lambda}\) and \( \boldsymbol{I}_l\), and set \(\boldsymbol{I}=(\boldsymbol{I}_1^\T,\dots,\boldsymbol{I}_L^\T)\). Explicitly, we have
\begin{align*}
\text{cov}(Z_{d_1l_1},Z_{d_2l_2}) &= \boldsymbol{\Lambda}_{d_1 + (D - 1)(l_1 - 1), d_2 + (D - 1)(l_2 - 1)},\\
&= \frac{1}{2}\left(\frac{l_1}{l_2}\right)^{1/2}(1 + \delta_{d_1d_2}),\\
I_{dl} &= \frac{ln}{2\sigma_e^2},
\end{align*}
for any $d,d_1,d_2 \in \{1,\dots,D-1\}$ and $l_1 \le l_2$, $l,l_1,l_2 \in \{1,\dots,L\}$.

    \section*{APPENDIX B: SMALL SAMPLE SIZE PERFORMANCE} \label{appB}

For small sample sizes, simulation is required to accurately determine a designs performance. Since crossover trials are routinely conducted with small sample sizes, we here explore the impact this has upon the familywise error rate under the global null hypothesis.

We determined a design corresponding to the TOMADO example that would require only 12 patients in each of two stages: the smallest allowable maximum sample size for a group sequential crossover trial with $D=4$ treatments initially, given our restrictions on $n$. Taking $\Delta=0$ as an example, a trial with $n=12$ and $L=2$ with $f_1=0.768, f_2=2.036$ and $e_1=2.879$ to 3 decimal places, would using our multivariate normal calculations have a maximal familywise error rate of $\alpha=0.05$ under the global null hypothesis, and $\beta=0.2$ for $\delta=2.2$ when $\sigma_e^2=6.51$.

Ten-thousand of these trials were simulated in order to ascertain the true probability of rejecting $H_{0d}$ for some $d=1,2,3$, when $\boldsymbol{\tau}=(0,0,0)^\T$. For simplicity, $\pi_j$ was set to 0 for $j=2,\dots,D$, and $\mu_0$ was set to 0. Incorporating non-zero period effects however would not be expected to greatly effect the results.

Whitehead et al. (2009) proposed a quantile substitution procedure for adapting the boundaries of a sequential trial to be more suitable to the case of unknown variance. We additionally considered employing this procedure. Given there is no consensus on how to determine the degrees of freedom when analysing using linear mixed models, we took the degrees of freedom at any analysis to be the classical decomposition of degrees of freedom in balanced, multilevel ANOVA designs (Pinheiro and Bates, 2009). Moreover, we also assessed the performance of the sequential design when the linear mixed model was fitted through either maximum likelihood or restricted error maximum likelihood estimation. Therefore in total these simulations were performed for each of four possible analysis procedures: maximum likelihood or restricted error maximum likelihood estimation, with or without boundary adjustment through quantile substitution.

Thus, for each simulated study, patient response data for each stage $l$ was randomly generated according to the distribution implied by their allocated treatment sequence (assigned according to the rules of the trial design), using the function \texttt{rmvnorm} (Genz et al., 2016) in R. The between person variance was set to $\sigma_b^2=10.12$; the value ascertained in the final analysis of the TOMADO trial data. Following this, our linear mixed model was fitted on all accumulated data (with either maximum likelihood or restricted error maximum likelihood estimation according to the particular analysis procedure being considered) and $Z_{dl}=\hat{\tau_{dl}}\hat{I}_{dl}^{1/2}$ determined for $d=1,2,3$, where $\hat{I}_{dl}^{1/2}$ is the observed Fisher information for $\hat{\tau_{dl}}$. Then, each $Z_{dl}$ was compared to $e_l$ and $f_l$ and our stopping rules applied (with $e_l$ and $f_l$ adjusted using quantile substitution if the analysis procedure under consideration so dictated). If for some $d=1,2,3$, $f_l \le Z_{dl} < e_l$, the trial proceeded to the following stage and the process was repeated. In each instance, simulations in which $H_{0d}$ was rejected for some $d=1,2,3$ were recorded in order to ascertain true rejection rates.

The performance of these procedures is displayed in Table~\ref{tab2}. We observe that when maximum likelihood estimation is utilised and the boundaries are not adjusted using the procedure of Whitehead et al. (2009), there is substantial inflation in the familywise error rate under the global null hypothesis to 0.077. However, when restricted error maximum likelihood estimation is used, there is only negligible inflation if adjustment of the boundaries is employed. A program to perform this analysis is available. Its use is detailed in Appendix D.

\begin{table*}
	\centering
	\caption{Performance of the small sample size group sequential crossover trial design under four analysis procedures (ML = Maximum Likelihood, REML = Restricted Error Maximum Likelihood). Specifically, \(\text{pr}\left(\text{Reject } H_{0d} \text{ for some } d\mid\boldsymbol{\tau}=\boldsymbol{0}\right)\) is shown for each procedure to 3 decimal places based on 10000 trial simulations}
	\begin{footnotesize}
		\begin{tabular}{lrrr}
			\hline
			Procedure & Estimation & Boundary Adjustment & \(\text{pr}\left(\text{Reject } H_{0d} \text{ for some } d\mid\boldsymbol{\tau}=\boldsymbol{0}\right)\)\\[5pt]
			Procedure 1 & ML & No & 0.077 \\
			Procedure 2 & ML & Yes & 0.062 \\
			Procedure 3 & REML & No & 0.055 \\
			Procedure 4 & REML & Yes & 0.051 \\
			\hline
		\end{tabular}
	\end{footnotesize}
	\label{tab2}
\end{table*}

    \section*{APPENDIX C: TECHNICAL PROOFS} \label{appC}
    
    \begin{lemma}
    	\label{lemma1} Element \(pq\) of \(\boldsymbol{\Sigma}_r^{-1}\) is given by
    	\begin{equation} \label{eq:eq3}
    	\boldsymbol{\Sigma}_{rpq}^{-1} = \frac{1}{\sigma_e^2(\sigma_e^2 + r\sigma_b^2)}\{(\sigma_e^2 + r\sigma_b^2)\delta_{pq} - \sigma_b^2\}.
    	\end{equation}
    \end{lemma}
    \begin{proof}
    	We demonstrate this by verifying \( \boldsymbol{\Sigma}_{rps} \boldsymbol{\Sigma}_{rsq}^{-1} =
    	\delta_{pq} \). From the chosen covariance for a patient's responses, we have that element \(pq\) of \(\boldsymbol{\Sigma}_r\) is given by
    	\begin{equation} \label{eq:eq4}
    	\boldsymbol{\Sigma}_{rpq} = \sigma_b^2\boldsymbol{Z}_{rps}\boldsymbol{Z}_{rsq}^\T + \sigma_e^2\delta_{pq} = \sigma_b^2 + \sigma_e^2\delta_{pq},
    	\end{equation}
    	where \(\boldsymbol{Z}_{rpq}\), \((p=1;\ q=1,\dots,r)\), is the \(pq\)th element of \(\boldsymbol{Z}_r\); the random effects design matrix for a single individual when there are \(r\) treatments remaining. Then
    	\begin{normalsize}
    		\begin{align*}
    		\boldsymbol{\Sigma}_{rps} \boldsymbol{\Sigma}_{rsq}^{-1} &= \frac{1}{\sigma_e^2(\sigma_e^2 + r\sigma_b^2)} \sum_{s=1}^r \left\{ \sigma_b^2 + \sigma_e^2 \delta_{ps} \right\} \left\{ (\sigma_e^2 + r \sigma_b^2) \delta_{sq} - \sigma_b^2 \right\},\\
    		&= \frac{1}{\sigma_e^2(\sigma_e^2 + r\sigma_b^2)}\sum_{s=1}^r \left\{\sigma_b^2(\sigma_e^2 + r\sigma_b^2)\delta_{sq} - \sigma_b^4 + \sigma_e^2(\sigma_e^2 + r\sigma_b^2)\delta_{pq}\right.\\
    		& \qquad\qquad \qquad \qquad \qquad \qquad \left. - \sigma_e^2\sigma_b^2\delta_{ps} \right\},\\
    		&= \frac{1}{\sigma_e^2(\sigma_e^2 + r\sigma_b^2)}\left\{\sigma_b^2(\sigma_e^2 + r\sigma_b^2) - r\sigma_b^4 + \sigma_e^2(\sigma_e^2 + r\sigma_b^2)\delta_{pq} - \sigma_e^2\sigma_b^2 \right\},\\
    		&= \delta_{pq},
    		\end{align*}
    	\end{normalsize}
    	using \(\sum_s \delta_{sq} = \sum_s \delta_{ps} = 1\).
    \end{proof}
    
    \begin{lemma}
    	\label{lemma2} Take the vector of fixed effects \( \boldsymbol{\beta} \) to be
    	\[ \boldsymbol{\beta} = (\mu_0,\pi_2,\dots,\pi_D,\tau_2,\dots,\tau_{D-1})^\T. \]
    	Then we have the following result
    	\begin{tiny}
    		\begin{equation} \label{eq:eq5}
    		\sum_{i=1}^{\left|S_r\right|} \frac{n}{\left|S_r\right|} \boldsymbol{X}_{s_{ri}}^\T \boldsymbol{\Sigma}_r^{-1} \boldsymbol{X}_{s_{ri}} = n \begin{pmatrix} A & \boldsymbol{B}^\T & \boldsymbol{0}_{1,D-r} & \boldsymbol{B}^\T & \boldsymbol{0}_{1,D-r} \\[0.3em]
    		\boldsymbol{B} & \boldsymbol{C} & \boldsymbol{0}_{r-1,D-r} & \boldsymbol{E} & \boldsymbol{0}_{r-1,D-r} \\[0.3em]
    		\boldsymbol{0}_{D-r,1} & \boldsymbol{0}_{D-r,r-1} & \boldsymbol{0}_{D-r,D-r} & \boldsymbol{0}_{D-r,r-1} & \boldsymbol{0}_{D-r,D-r} \\[0.3em]
    		\boldsymbol{B} & \boldsymbol{E} & \boldsymbol{0}_{r-1,D-r} & \boldsymbol{C} & \boldsymbol{0}_{r-1,D-r} \\[0.3em]
    		\boldsymbol{0}_{D-r,1} & \boldsymbol{0}_{D-r,r-1} & \boldsymbol{0}_{D-r,D-r} & \boldsymbol{0}_{D-r,r-1} & \boldsymbol{0}_{D-r,D-r}
    		\end{pmatrix},
    		\end{equation}
    	\end{tiny}
    	where \(\boldsymbol{0}_{m,n}\) is a matrix of zeroes of dimension \(m \times n\), and
    	\begin{small}
    		\begin{align*}
    		A &= \frac{r\sigma_e^2}{\sigma_e^2(\sigma_e^2 + r\sigma_b^2)}, \\
    		\boldsymbol{B}_{pq} &= \frac{r\sigma_e^2}{\sigma_e^2(\sigma_e^2 + r\sigma_b^2)} & (p&=1,\dots,r-1;\ q=1), \\
    		\boldsymbol{C}_{pq} &= \frac{1}{\sigma_e^2(\sigma_e^2 + r\sigma_b^2)}\{(\sigma_e^2 + r\sigma_b^2)\delta_{pq} - \sigma_b^2\} & (p&=1,\dots,r-1;\ q=1,\dots,r-1), \\
    		\boldsymbol{E}_{pq} &= \frac{1}{r}\frac{\sigma_e^2}{\sigma_e^2(\sigma_e^2 + r\sigma_b^2)} & (p&=1,\dots,r-1;\ q=1,\dots,r-1).
    		\end{align*}
    	\end{small}
    \end{lemma}
    \begin{proof} Denote the columns of \(\boldsymbol{X}_{s_{ri}}\) by
    	\[ \boldsymbol{X}_{s_{ri}} = \begin{pmatrix}
    	\boldsymbol{1}_r & \boldsymbol{\Pi}_{2r} & \dots & \boldsymbol{\Pi}_{Dr} & \boldsymbol{T}_{1s_{ri}} & \dots & \boldsymbol{T}_{D-1 s_{ri}}
    	\end{pmatrix}. \]
    	Thus \( \boldsymbol{\Pi}_{jr} \) is the column corresponding to the period effect \( \pi_j \), \( \boldsymbol{T}_{d s_{ri}} \) to the treatment effect \(\tau_d\), and \(\boldsymbol{1}_r\) to the intercept \(\mu_0\). Using this representation, and Lemma~\ref{lemma1}, we have
    	
    	\begin{footnotesize}
    		\begin{align*}
    		\boldsymbol{X}_{s_{ri}}^\T\boldsymbol{\Sigma}_r^{-1}\boldsymbol{X}_{s_{ri}} &= \begin{pmatrix} \boldsymbol{1}_r^\T \\[0.3em]
    		\boldsymbol{\Pi}_{2r}^\T \\[0.3em]
    		\vdots \\[0.3em]
    		\boldsymbol{T}_{D-1 s_{ri}}^\T
    		\end{pmatrix} \boldsymbol{\Sigma}_r^{-1}
    		\begin{pmatrix} \boldsymbol{1}_r & \boldsymbol{\Pi}_{2r} & \dots
    		& \boldsymbol{T}_{D-1 s_{ri}}
    		\end{pmatrix} \\
    		&= \begin{pmatrix}
    		\boldsymbol{1}_r^\T\boldsymbol{\Sigma}_r^{-1}\boldsymbol{1}_r & \boldsymbol{1}_r^\T\boldsymbol{\Sigma}_r^{-1}\boldsymbol{\Pi}_{2r} & \dots & \boldsymbol{1}_r^\T\boldsymbol{\Sigma}_r^{-1}\boldsymbol{T}_{D-1 s_{ri}} \\[0.3em]
    		\boldsymbol{\Pi}_{2r}^\T\boldsymbol{\Sigma}_r^{-1}\boldsymbol{1}_r & \boldsymbol{\Pi}_{2r}^\T\boldsymbol{\Sigma}_r^{-1}\boldsymbol{\Pi}_{2r} & \dots & \boldsymbol{\Pi}_{2r}^\T\boldsymbol{\Sigma}_r^{-1}\boldsymbol{T}_{D-1 s_{ri}} \\[0.3em]
    		\vdots & \vdots & \ddots & \vdots \\[0.3em]
    		\boldsymbol{T}_{D-1 s_{ri}}^\T\boldsymbol{\Sigma}_r^{-1}\boldsymbol{1}_r & \boldsymbol{T}_{D-1 s_{ri}}^\T\boldsymbol{\Sigma}_r^{-1}\boldsymbol{\Pi}_{2r} & \dots & \boldsymbol{T}_{D-1 s_{ri}}^\T\boldsymbol{\Sigma}_r^{-1}\boldsymbol{T}_{D-1 s_{ri}}
    		\end{pmatrix}.
    		\end{align*}
    	\end{footnotesize}
    	Here, \( \text{dim}⁡(\boldsymbol{\Pi}_{jr}) = \text{dim}(\boldsymbol{1}_r) = \text{dim}⁡(\boldsymbol{T}_{j
    		s_ri}) = r \times 1 \), for all \(i\) and \(j\). Therefore
    	\(\boldsymbol{1}_r^\T\boldsymbol{\Sigma}_r^{-1}\boldsymbol{1}_r\), \(\boldsymbol{\Pi}_{jr}^\T\boldsymbol{\Sigma}_r^{-1}\boldsymbol{1}_r\),
    	\(\boldsymbol{T}_{js_{ri}}^\T\boldsymbol{\Sigma}_r^{-1}\boldsymbol{1}_r\), \(\boldsymbol{1}_r^\T\boldsymbol{\Sigma}_r^{-1}\boldsymbol{\Pi}_{jr}\),
    	\(\boldsymbol{1}_r^\T\boldsymbol{\Sigma}_r^{-1}\boldsymbol{T}_{js_{ri}}\),
    	\(\boldsymbol{\Pi}_{jr}^\T\boldsymbol{\Sigma}_r^{-1}\boldsymbol{\Pi}_{kr}\),
    	\(\boldsymbol{T}_{js_{ri}}^\T\boldsymbol{\Sigma}_r^{-1}\boldsymbol{T}_{ks_{ri}}\) and
    	\(\boldsymbol{\Pi}_{jr}^\T\boldsymbol{\Sigma}_r^{-1}\boldsymbol{T}_{ks_{ri}}\) are scalars for all \(i\),
    	\(j\) and \(k\).
    	
    	By the definition of being in period \(j\), the \(v\)th element of
    	\(\boldsymbol{\Pi}_{jr}\) is given by
    	\[ \boldsymbol{\Pi}_{jrv} = \delta_{jv} .\]
    	Whilst, since complete block design sequences are used, the \(v\)th element
    	of \(\boldsymbol{T}_{js_{ri}}\) is given by
    	\[ \boldsymbol{T}_{js_{ri}v} = \delta_{1\mathbb{I}_{js_{ri}v}}, \]
    	where
    	\[ \mathbb{I}_{js_{ri}v} = \left\{
    	\begin{array}{l l}
    	1 & \text{An individual on sequence } s_{ri} \text{ receives treatment } j \text{ in period } v,\\
    	0 & \text{otherwise}.
    	\end{array}\right. \]
    	We denote this non-zero element, if it exists, by \(t_j\).
    	
    	Now from the symmetry present in \(\boldsymbol{\Sigma}_r^{-1}\), using Lemma~\ref{lemma1} we have
    	\[ \sum_{j,k} \boldsymbol{\Sigma}_{rjk}^{-1} = \frac{r\sigma_e^2}{\sigma_e^2(\sigma_e^2 + r\sigma_b^2)}, \]
    	and
    	\[ \sum_{j} \boldsymbol{\Sigma}_{rjk}^{-1} = \sum_{k}\boldsymbol{\Sigma}_{rjk}^{-1} = \frac{\sigma_e^2}{\sigma_e^2(\sigma_e^2 + r\sigma_b^2)}, \]
    	for all \(j\) and \(k\). Therefore, we can determine the form of each of the
    	scalar elements in the matrix above as follows
    	\begin{footnotesize}
    		\begin{align*}
    		\boldsymbol{1}_r^\T\boldsymbol{\Sigma}_r^{-1}\boldsymbol{1}_r &= \sum_{j,k} \boldsymbol{\Sigma}_{rjk}^{-1} = \frac{r\sigma_e^2}{\sigma_e^2(\sigma_e^2 + r\sigma_b^2)}, \\
    		\boldsymbol{\Pi}_{jr}^\T\boldsymbol{\Sigma}_r^{-1}\boldsymbol{1}_r &= \boldsymbol{1}_r^\T\boldsymbol{\Sigma}_r^{-1}\boldsymbol{\Pi}_{jr} = \sum_{k} \boldsymbol{\Sigma}_{rjk}^{-1}\mathbb{I}_{\{j \le r\}} = \frac{\sigma_e^2}{\sigma_e^2(\sigma_e^2 + r\sigma_b^2)}\mathbb{I}_{\{j \le r\}},\\
    		\boldsymbol{T}_{js_{ri}}^\T\boldsymbol{\Sigma}_r^{-1}\boldsymbol{1}_r &= \boldsymbol{1}_r^\T\boldsymbol{\Sigma}_r^{-1}\boldsymbol{T}_{js_{ri}} = \sum_{k} \boldsymbol{\Sigma}_{rt_jk}^{-1}\mathbb{I}_{\{j \le r-1\}} = \frac{\sigma_e^2}{\sigma_e^2(\sigma_e^2 + r\sigma_b^2)}\mathbb{I}_{\{j \le r-1\}},\\
    		\boldsymbol{\Pi}_{jr}^\T\boldsymbol{\Sigma}_r^{-1}\boldsymbol{\Pi}_{kr} &= \boldsymbol{\Sigma}_{rjk}^{-1}\mathbb{I}_{\{j \le r\}}\mathbb{I}_{\{k \le r\}} = \left\{\frac{(\sigma_e^2 + r\sigma_b^2)\mathbb{I}_{\{j=k\}} - \sigma_b^2}{\sigma_e^2(\sigma_e^2 + r\sigma_b^2)}\right\}\mathbb{I}_{\{j \le r\}}\mathbb{I}_{\{k \le r\}},\\
    		\boldsymbol{T}_{js_{ri}}^\T\boldsymbol{\Sigma}_r^{-1}\boldsymbol{T}_{ks_{ri}} &= \boldsymbol{\Sigma}_{rt_jt_k}^{-1}\mathbb{I}_{\{j \le r-1\}}\mathbb{I}_{\{k \le r-1\}} = \left\{\frac{(\sigma_e^2 + r\sigma_b^2)\mathbb{I}_{\{j=k\}} - \sigma_b^2}{\sigma_e^2(\sigma_e^2 + r\sigma_b^2)}\right\}\mathbb{I}_{\{j \le r-1\}}\mathbb{I}_{\{k \le r-1\}},\\
    		\boldsymbol{\Pi}_{jr}^\T\boldsymbol{\Sigma}_r^{-1}\boldsymbol{T}_{ks_{ri}} &= \boldsymbol{\Sigma}_{rjt_k}^{-1}\mathbb{I}_{\{j \le r\}}\mathbb{I}_{\{k \le r-1\}} = \left\{\frac{(\sigma_e^2 + r\sigma_b^2)\mathbb{I}_{\{j=t_k\}} - \sigma_b^2}{\sigma_e^2(\sigma_e^2 + r\sigma_b^2)}\right\}\mathbb{I}_{\{j \le r\}}\mathbb{I}_{\{k \le r-1\}}.
    		\end{align*}
    	\end{footnotesize}
    	
    	As a final step we must compute the sum across sequences, i.e. over the index \(i\). We can see instantly that we have confirmed the elements proposed to be 0 in \(\sum_{i=1}^{\left|S_r\right|} n\boldsymbol{X}_{s_{ri}}^\T \boldsymbol{\Sigma}_r^{-1} \boldsymbol{X}_{s_{ri}}/\left|S_r\right|\) are indeed so, and we therefore need only concentrate on the non-zero terms suggested; \(A\), \(B\), \(C\) and \(E\).
    	
    	However, other than \( \boldsymbol{\Pi}_{jr}^\T\boldsymbol{\Sigma}_r^{-1}\boldsymbol{T}_{ks_{ri}} \), all the
    	elements above have been identified as independent of sequence \(s_{ri}\).
    	Therefore computing the sum over the \(s_{ri} \in S_r \) can be done easily,
    	and gives the forms proposed for \(A\), \(B\) and \(C\) in the statement of
    	the Lemma immediately, on multiplying through by \(n/\left|S_r\right|\). But, by our imposed constraint that sequences be balanced for period we can also sum over the \( \boldsymbol{\Pi}_{jr}^\T\boldsymbol{\Sigma}_r^{-1}\boldsymbol{T}_{ks_{ri}} \)
    	
    	\begin{footnotesize}
    		\begin{align*}
    		\sum_{i=1}^{\left|S_r\right|} \frac{n}{\left|S_r\right|} \boldsymbol{\Pi}_{jr}^\T\boldsymbol{\Sigma}_r^{-1}\boldsymbol{\mathrm{T}}_{ks_{ri}} &= \frac{n}{\left|S_r\right|} \left( \frac{\left|S_r\right|}{r}\left[ \frac{1}{\sigma_e^2(\sigma_e^2 + r\sigma_b^2)} \left\{ \sigma_e^2 + (r-1)\sigma_b^2 \right\} \right] + \right.\\
    		& \qquad \qquad \left. \left\{1 - \frac{\left|S_r\right|}{r} \right\} \left\{\frac{-\sigma_b^2}{\sigma_e^2(\sigma_e^2 + r\sigma_b^2)} \right\} \right),\\
    		&= \frac{n}{\left|S_r\right|}\left\{\frac{\left|S_r\right|}{r}\frac{\sigma_e^2}{\sigma_e^2(\sigma_e^2 + r\sigma_b^2)}\right\},\\
    		&= n\frac{1}{r}\frac{\sigma_e^2}{\sigma_e^2(\sigma_e^2 + r\sigma_b^2)},
    		\end{align*}
    	\end{footnotesize}
    	since exactly \(\left|S_r\right|/r\) patients receive each treatment at each
    	time period. This confirms the form proposed for \(E\), and the proof is
    	complete.
    \end{proof}
    
    \begin{thm}
    	\label{theorem1} (Theorem~\ref{thm1} from Section~\ref{s:model}) Let \( \boldsymbol{\beta} = (\mu_0,\pi_2,\dots,\pi_D,\tau_1,\dots,\tau_{D-1})^\T \). Consider an analysis to be performed after some number of stages $l$. Then 
    	\begin{enumerate}
    		\item We have
    		\begin{footnotesize}
    			\begin{align}
    			\text{cov}(\hat{\boldsymbol{\beta}}_l,\hat{\boldsymbol{\beta}}_l \mid L_{l1},\dots,L_{lD-1}=0,L_{lD}=l) &= \left( \frac{ln}{\left|S_D\right|} \sum_{i=1}^{\left|S_D\right|} \boldsymbol{X}_{s_{Di}}^\T \boldsymbol{\Sigma}_D^{-1} \boldsymbol{X}_{s_{Di}} \right)^{-1},\nonumber\\
    			&= \frac{1}{ln} \begin{pmatrix} F & \boldsymbol{G}^\T & \boldsymbol{G}^\T \\[0.3em]
    			\boldsymbol{G} & \boldsymbol{H} & \boldsymbol{0}_{D-1,D-1} \\[0.3em]
    			\boldsymbol{G} & \boldsymbol{0}_{D-1,D-1} & \boldsymbol{H}
    			\end{pmatrix},\label{eq:eq2}
    			\end{align}
    		\end{footnotesize}
    		where
    		\begin{align*}
    		F &= \sigma_b^2 + \frac{2D-1}{D}\sigma_e^2,\\
    		\boldsymbol{G}_{pq} &= -\sigma_e^2 & (p&=1,\dots,D-1;\ q=1), \\
    		\boldsymbol{H}_{pq} &= \sigma_e^2(1 + \delta_{pq}) & (p&=1,\dots,D-1;\ q=1,\dots,D-1).
    		\end{align*}
    		\item If $q \ge 2$ is the largest integer such that $L_{lr} = 0$ for $r=1,\dots,q-1$, then the covariance of the estimates of the fixed effects \(\hat{\pi}_{2l},\dots,\hat{\pi}_{ql}, \hat{\tau}_{1l},\dots,\hat{\tau}_{q-1 l}\) is identical to that it would be for \(L_{l1}=\dots=L_{lD-1}=0\). Moreover, the covariance between the estimates of \(\hat{\pi}_{2l},\dots,\hat{\pi}_{ql}, \hat{\tau}_{1l},\dots,\hat{\tau}_{q-1 l}\) and the estimates of \(\hat{\pi}_{q+1l},\dots,\hat{\pi}_{Dl}, \hat{\tau}_{ql},\dots,\hat{\tau}_{D-1 l}\) is also identical to that it would be for \(L_{l1}=\dots=L_{lD-1}=0\).
    	\end{enumerate}
    \end{thm}
    \begin{proof}
    	1. We begin with the result for the case \(L_{l1}=\dots=L_{1D-1}=0\). We demonstrate this by confirming
    	\begin{scriptsize}
    		\[
    		\text{cov}(\hat{\boldsymbol{\beta}}_l,\hat{\boldsymbol{\beta}}_l \mid L_{l1},\dots,L_{lD-1}=0,L_{lD}=l)^{-1}\text{cov}(\hat{\boldsymbol{\beta}}_l,\hat{\boldsymbol{\beta}}_l \mid L_{l1},\dots,L_{lD-1}=0,L_{lD}=l)
    		= \boldsymbol{I}_{2D-1}.\]
    	\end{scriptsize}
    	By Lemma~\ref{lemma2} we know that
    	\[ \text{cov}(\hat{\boldsymbol{\beta}}_l,\hat{\boldsymbol{\beta}}_l \mid L_{l1},\dots,L_{lD-1}=0,L_{lD}=l)^{-1} = ln\begin{pmatrix}
    	A & \boldsymbol{B}^\T & \boldsymbol{B}^\T \\[0.3em]
    	\boldsymbol{B} & \boldsymbol{C} & \boldsymbol{E} \\[0.3em]
    	\boldsymbol{B} & \boldsymbol{E} & \boldsymbol{C}
    	\end{pmatrix}, \]
    	for
    	\begin{footnotesize}
    		\begin{align*}
    		A &= \frac{D\sigma_e^2}{\sigma_e^2(\sigma_e^2 + D\sigma_b^2)},\\
    		\boldsymbol{B} &= \frac{D\sigma_e^2}{\sigma_e^2(\sigma_e^2 + D\sigma_b^2)}\begin{pmatrix}1\\[0.3em]
    		\vdots\\[0.3em]
    		1
    		\end{pmatrix},\\
    		\boldsymbol{C} &= \frac{1}{\sigma_e^2(\sigma_e^2 + D\sigma_b^2)}
    		\begin{pmatrix}
    		\sigma_e^2 + (D-1)\sigma_b^2 & -\sigma_b^2 & \dots & -\sigma_b^2 \\[0.3em]
    		-\sigma_b^2 & \sigma_e^2 + (D-1)\sigma_b^2 & \ddots & \vdots \\[0.3em]
    		\vdots & \ddots & \ddots & -\sigma_b^2 \\[0.3em]
    		-\sigma_b^2 & \dots & -\sigma_b^2 & \sigma_e^2 + (D-1)\sigma_b^2 &
    		\end{pmatrix},\\
    		\boldsymbol{E} &= \frac{1}{D}\frac{\sigma_e^2}{\sigma_e^2(\sigma_e^2 + D\sigma_b^2)}
    		\begin{pmatrix}
    		1 & \dots & 1 \\[0.3em]
    		\vdots & \ddots & \vdots \\[0.3em]
    		1 & \dots & 1
    		\end{pmatrix},
    		\end{align*}
    	\end{footnotesize}
    	with
    	\begin{align*}
    	\text{dim}(\boldsymbol{B}) &= (D-1) \times 1,\\
    	\text{dim}(\boldsymbol{C}) &= (D-1) \times (D-1),\\
    	\text{dim}(\boldsymbol{E}) &= (D-1) \times (D-1).
    	\end{align*}
    	Thus we must show
    	\[ \begin{pmatrix}
    	\boldsymbol{A} & \boldsymbol{B}^\T & \boldsymbol{B}^\T \\[0.3em]
    	\boldsymbol{B} & \boldsymbol{C} & \boldsymbol{E} \\[0.3em]
    	\boldsymbol{B} & \boldsymbol{E} & \boldsymbol{C}
    	\end{pmatrix} \begin{pmatrix}
    	\boldsymbol{F} & \boldsymbol{G}^\T & \boldsymbol{G}^\T \\[0.3em]
    	\boldsymbol{G} & \boldsymbol{H} & \boldsymbol{0}_{D-1,D-1} \\[0.3em]
    	\boldsymbol{G} & \boldsymbol{0}_{D-1,D-1} & \boldsymbol{H}
    	\end{pmatrix}=\boldsymbol{I}_{2D-1},\]
    	or on expanding
    	\begin{align*}
    	\boldsymbol{A}\boldsymbol{F} +2\boldsymbol{B}^\T \boldsymbol{G} &= 1,\\
    	\boldsymbol{B}\boldsymbol{G}^\T + \boldsymbol{C}\boldsymbol{H} &= \boldsymbol{I}_{D-1},\\
    	\boldsymbol{B}\boldsymbol{G}^\T + \boldsymbol{E}\boldsymbol{H} &= \boldsymbol{0}_{D-1,D-1},\\
    	\boldsymbol{B}\boldsymbol{F} + \boldsymbol{C}\boldsymbol{G} + \boldsymbol{E}\boldsymbol{G} &= \boldsymbol{0}_{D-1,1},\\
    	\boldsymbol{A}\boldsymbol{G}^\T + \boldsymbol{B}^\T \boldsymbol{H} &= \boldsymbol{0}_{1,D-1}.
    	\end{align*}
    	Now
    	\begin{tiny}
    		\begin{align*}
    		\phantom{\boldsymbol{B}\boldsymbol{F} + \boldsymbol{C}\boldsymbol{G} + \boldsymbol{E}\boldsymbol{G}}
    		&\begin{aligned}
    		\mathllap{\boldsymbol{A}\boldsymbol{F} + 2\boldsymbol{B}^\T \boldsymbol{G}} &= \left( \frac{D\sigma_e^2}{\sigma_e^2(\sigma_e^2 + D\sigma_b^2)}\right)\left(\sigma_b^2 + \frac{2D-1}{D}\sigma_e^2 \right) \\
    		&\qquad  + 2\frac{\sigma_e^2}{\sigma_e^2(\sigma_e^2 + D\sigma_b^2)}\begin{pmatrix}1 \\[0.3em] \vdots \\[0.3em] 1 \end{pmatrix}^\T(-\sigma_e^2)\begin{pmatrix}1 \\[0.3em] \vdots \\[0.3em] 1 \end{pmatrix},
    		\end{aligned}\\
    		&\begin{aligned}
    		\mathllap{} &= \left( \frac{D\sigma_e^2}{\sigma_e^2(\sigma_e^2 + D\sigma_b^2)}\right)\left(\sigma_b^2 + \frac{2D-1}{D}\sigma_e^2 \right)\\
    		&\qquad - 2(D-1)\sigma_e^2\frac{\sigma_e^2}{\sigma_e^2(\sigma_e^2 + D\sigma_b^2)},
    		\end{aligned}\\
    		&= 1,\displaybreak[0] \\
    		&\begin{aligned}
    		\mathllap{\boldsymbol{B}\boldsymbol{G}^\T + \boldsymbol{C}\boldsymbol{H}} &= \frac{\sigma_e^2}{\sigma_e^2(\sigma_e^2 + D\sigma_b^2)}\begin{pmatrix}1 \\[0.3em] \vdots \\[0.3em] 1 \end{pmatrix}(-\sigma_e^2) \begin{pmatrix}1 \\[0.3em] \vdots \\[0.3em] 1 \end{pmatrix}^\T \\
    		&\qquad + \frac{1}{\sigma_e^2(\sigma_e^2 + D\sigma_b^2)}\begin{pmatrix} \sigma_e^2 + (D-1)\sigma_b^2 & -\sigma_b^2 & \dots & -\sigma_b^2 \\[0.3em] -\sigma_b^2 & \sigma_e^2 + (D-1)\sigma_b^2 & \ddots & \vdots \\[0.3em] \vdots & \ddots & \ddots & -\sigma_b^2 \\[0.3em] -\sigma_b^2 & \dots & -\sigma_b^2 & \sigma_e^2 + (D-1)\sigma_b^2 \end{pmatrix} \times \\
    		&\qquad \qquad \times (\sigma_e^2)\begin{pmatrix} 2 & 1 & \dots & 1 \\[0.3em] 1 & \ddots & \ddots & \vdots \\[0.3em] \vdots & \ddots & \ddots & 1 \\[0.3em] 1 & \dots & 1 & 2 \end{pmatrix},
    		\end{aligned}\displaybreak[0]\\
    		&\begin{aligned}
    		\mathllap{} &= -\frac{\sigma_e^4}{\sigma_e^2(\sigma_e^2 + D\sigma_b^2)}\begin{pmatrix} 1 & 1 & \dots & 1\\[0.3em]
    		1 & \ddots & \ddots & \vdots \\[0.3em]
    		\vdots & \ddots & \ddots & 1 \\[0.3em]
    		1 & \dots & 1 & 1
    		\end{pmatrix}\\
    		&\qquad + \frac{\sigma_e^2}{\sigma_e^2(\sigma_e^2 + D\sigma_b^2)}
    		\begin{pmatrix} 2\sigma_e^2 & \sigma_e^2 & \dots &
    		\sigma_e^2 \\[0.3em]
    		\sigma_e^2 & 2\sigma_e^2 & \ddots & \vdots \\[0.3em]
    		\vdots & \ddots & \ddots & \sigma_e^2 \\[0.3em]
    		\sigma_e^2 & \dots & \sigma_e^2 & 2\sigma_e^2
    		\end{pmatrix},
    		\end{aligned}\\
    		&=\boldsymbol{I}_{D-1},\displaybreak[0]\\
    		&\begin{aligned}
    		\mathllap{\boldsymbol{B}\boldsymbol{G}^\T + \boldsymbol{E}\boldsymbol{H}} &= \frac{\sigma_e^2}{\sigma_e^2(\sigma_e^2 + D\sigma_b^2)}\begin{pmatrix}1 \\[0.3em] \vdots \\[0.3em] 1 \end{pmatrix}(-\sigma_e^2) \begin{pmatrix}1 \\[0.3em] \vdots \\[0.3em] 1 \end{pmatrix}^\T\\
    		&\qquad + \frac{1}{D}\frac{\sigma_e^2}{\sigma_e^2(\sigma_e^2 + D\sigma_b^2)}\begin{pmatrix} 1 & 1 & \dots & 1 \\[0.3em] 1 & \ddots & \ddots & \vdots \\[0.3em] \vdots & \ddots & \ddots & 1 \\[0.3em] 1 & \dots & 1 & 1 \end{pmatrix} (\sigma_e^2)\begin{pmatrix} 2 & 1 & \dots & 1 \\[0.3em] 1 & \ddots & \ddots & \vdots \\[0.3em] \vdots & \ddots & \ddots & 1 \\[0.3em] 1 & \dots & 1 & 2 \end{pmatrix},
    		\end{aligned}\displaybreak[0]\\
    		&\begin{aligned}
    		\mathllap{} &= -\frac{\sigma_e^4}{\sigma_e^2(\sigma_e^2 + D\sigma_b^2)}\begin{pmatrix} 1 & 1 & \dots & 1\\[0.3em]
    		1 & \ddots & \ddots & \vdots \\[0.3em]
    		\vdots & \ddots & \ddots & 1 \\[0.3em]
    		1 & \dots & 1 & 1
    		\end{pmatrix}\\
    		&\qquad + \frac{1}{D}\frac{\sigma_e^4}{\sigma_e^2(\sigma_e^2 + D\sigma_b^2)}
    		\begin{pmatrix} D & D & \dots & D \\[0.3em]
    		D & \ddots & \ddots & \vdots \\[0.3em]
    		\vdots & \ddots & \ddots & D \\[0.3em]
    		D & \ddots & D & D
    		\end{pmatrix},
    		\end{aligned}\\
    		&=\boldsymbol{0}_{D-1,D-1},\displaybreak[0]\\
    		&\begin{aligned}
    		\mathllap{\boldsymbol{B}\boldsymbol{F} + \boldsymbol{C}\boldsymbol{G} + \boldsymbol{E}\boldsymbol{G}} &= \frac{\sigma_e^2}{\sigma_e^2(\sigma_e^2 + D\sigma_b^2)}\begin{pmatrix} 1 \\[0.3em] 1 \\[0.3em] \vdots \\[0.3em] 1 \end{pmatrix} \left(\sigma_b^2 + \frac{2D-1}{D}\sigma_e^2 \right)\\
    		&\qquad + \frac{1}{\sigma_e^2(\sigma_e^2 + D\sigma_b^2)}\begin{pmatrix} \sigma_e^2 + (D-1)\sigma_b^2 & -\sigma_b^2 & \dots & -\sigma_b^2 \\[0.3em] -\sigma_b^2 & \sigma_e^2 + (D-1)\sigma_b^2 & \ddots & \vdots \\[0.3em] \vdots & \ddots & \ddots & -\sigma_b^2 \\[0.3em] -\sigma_b^2 & \dots & -\sigma_b^2 & \sigma_e^2 + (D-1)\sigma_b^2 \end{pmatrix} \times \\
    		&\qquad \qquad \times (\sigma_e^2)\begin{pmatrix} 1 \\[0.3em] 1 \\[0.3em] \vdots \\[0.3em] 1 \end{pmatrix} + \frac{1}{D}\frac{\sigma_e^2}{\sigma_e^2(\sigma_e^2 + D\sigma_b^2)}\begin{pmatrix} 1 & 1 & \dots & 1\\[0.3em]
    		1 & \ddots & \ddots & \vdots \\[0.3em]
    		\vdots & \ddots & \ddots & 1 \\[0.3em]
    		1 & \dots & 1 & 1
    		\end{pmatrix} (-\sigma_e^2) \begin{pmatrix} 1 \\[0.3em] 1 \\[0.3em] \vdots \\[0.3em] 1 \end{pmatrix},
    		\end{aligned}\displaybreak[0]\\
    		&\begin{aligned}
    		\mathllap{} &= \frac{\sigma_e^2}{\sigma_e^2(\sigma_e^2 + D\sigma_b^2)}\left(\sigma_b^2 + \frac{2D-1}{D}\sigma_e^2 \right)\begin{pmatrix} 1 \\[0.3em] 1 \\[0.3em] \vdots \\[0.3em] 1 \end{pmatrix}\\
    		&\qquad - \frac{\sigma_e^2}{\sigma_e^2(\sigma_e^2 + D\sigma_b^2)}(\sigma_e^2 + \sigma_b^2)\begin{pmatrix} 1 \\[0.3em] 1 \\[0.3em] \vdots \\[0.3em] 1 \end{pmatrix} - \frac{1}{D}\frac{\sigma_e^4}{\sigma_e^2(\sigma_e^2 + D\sigma_b^2)}(D-1)\begin{pmatrix} 1 \\[0.3em] 1 \\[0.3em] \vdots \\[0.3em] 1 \end{pmatrix},
    		\end{aligned}\\
    		&= \boldsymbol{0}_{D-1,1},\displaybreak[0]\\
    		&\begin{aligned}
    		\mathllap{\boldsymbol{A}\boldsymbol{G}^\T + \boldsymbol{B}^\T \boldsymbol{H}} &= \left\{\frac{D\sigma_e^2}{\sigma_e^2(\sigma_e^2 + D\sigma_b^2)} \right\}(-\sigma_e^2)\begin{pmatrix} 1 \\[0.3em] 1 \\[0.3em] \vdots \\[0.3em] 1 \end{pmatrix}^\T \\
    		&\qquad + \left\{\frac{\sigma_e^2}{\sigma_e^2(\sigma_e^2 + D\sigma_b^2)} \right\}\begin{pmatrix} 1 \\[0.3em] 1 \\[0.3em] \vdots \\[0.3em] 1 \end{pmatrix}^\T(\sigma_e^2) \begin{pmatrix} 2 & 1 & \dots & 1 \\[0.3em] 1 & \ddots & \ddots & \vdots \\[0.3em] \vdots & \ddots & \ddots & 1 \\[0.3em] 1 & \dots & 1 & 2 \end{pmatrix},
    		\end{aligned}\displaybreak[0]\\
    		&\begin{aligned}
    		\mathllap{} &= -\frac{D\sigma_e^4}{\sigma_e^2(\sigma_e^2 + D\sigma_b^2)} \begin{pmatrix} 1 \\ 1 \\ \vdots \\ 1 \end{pmatrix}^\T + \frac{\sigma_e^4}{\sigma_e^2(\sigma_e^2 + D\sigma_b^2)} D \begin{pmatrix} 1 \\ 1 \\ \vdots \\ 1 \end{pmatrix}^\T,
    		\end{aligned}\\
    		&= \boldsymbol{0}_{1,D-1},
    		\end{align*}
    	\end{tiny}
    	as required. Thus, the proposed matrix is indeed \(
    	\text{cov}(\hat{\boldsymbol{\beta}}_l,\hat{\boldsymbol{\beta}}_l \mid L_{l1},\dots,L_{lD-1}=0,L_{lD}=l) \).
    	
    	2. For the next part of the theorem, we re-order our vector of fixed effects such that
    	\[ \boldsymbol{\beta}=(\mu_0,\pi_2,\tau_1,\pi_3,\tau_2,\dots,\pi_D,\tau_{D-1})^\T,\]
    	and thus the ordering of the columns of the \(\boldsymbol{X}_{s_{ri}}\) is now
    	\[ \boldsymbol{X}_{s_{ri}} = \begin{pmatrix} \boldsymbol{1}_r & \boldsymbol{\Pi}_{2r} & \boldsymbol{T}_{1s_{ri}} & \boldsymbol{\Pi}_{3r} & \boldsymbol{T}_{2s_{ri}} & \dots & \boldsymbol{\Pi}_{Dr} & \boldsymbol{T}_{D-1 s_{ri}} \end{pmatrix}.\]
    	
    	We proceed by induction over the number of stages completed \(l\), for
    	general \(D\). Now, we assume that at the \(l\)th interim analysis, the
    	statement of the Theorem is true. Now, the covariance at this $l$th analysis is
    	\begin{footnotesize}
    		\begin{align*}
    		\text{cov}(\hat{\boldsymbol{\beta}}_l,\hat{\boldsymbol{\beta}}_l \mid L_{l1}=\dots=L_{lq-1}=0,L_{lq},\dots,L_{lD}) &= \left( \sum_{r=q}^{D} L_r \frac{n}{\left|S_r\right|} \sum_{i=1}^{\mid S_r \mid} \boldsymbol{X}^\T_{s_{ri}}\boldsymbol{\Sigma}_r^{-1}\boldsymbol{X}_{s_{ri}} \right)^{-1}\\
    		&=\boldsymbol{A}^{-1}.
    		\end{align*}
    	\end{footnotesize}
    	It is this specifically we assume follows the required condition. Additionally, we assume
    	that this covariance matrix can be computed, i.e. that \(\boldsymbol{A}\) is invertible.
    	We show that if we conduct another stage of the trial with \(t\) treatments
    	remaining, \(2 \le t \le q \), then the new covariance matrix
    	\begin{align*}
    	\text{cov}(\hat{\boldsymbol{\beta}}_{l+1},\hat{\boldsymbol{\beta}}_{l+1} \mid L_{l+11},\dots,L_{l+1D}) &= \left( \boldsymbol{A} + \frac{n}{\left|S_t\right|} \sum_{i=1}^{\mid S_t \mid} \boldsymbol{X}^\T_{s_{ti}}\boldsymbol{\Sigma}_t^{-1}\boldsymbol{X}_{s_{ti}} \right)^{-1}\\
    	&= \left(\boldsymbol{A} + \boldsymbol{B}\right)^{-1},
    	\end{align*}
    	has the required property for \(\hat{\pi}_{2l},\dots,\hat{\pi}_{tl}, \hat{\tau}_{1l},\dots,\hat{\tau}_{t-1 l}\) as well as for \(\hat{\pi}_{2l},\dots,\hat{\pi}_{tl},\) \(\hat{\tau}_{1l},\dots,\hat{\tau}_{t-1 l}\) and \(\hat{\pi}_{t+1l},\dots,\hat{\pi}_{Dl}\), \(\hat{\tau}_{tl},\dots,\hat{\tau}_{D-1 l}\).
    	Let
    	\begin{align*}
    	\boldsymbol{T}_l &= (\hat{\mu}_{0l},\hat{\pi}_{2l},\hat{\tau}_{1l},\dots,\hat{\pi}_{tl},\hat{\tau}_{t-1 l})^\T,\\
    	\boldsymbol{T}_l' &= (\hat{\pi}_{t+1 l},\hat{\tau}_{tl},\dots,\hat{\pi}_{Dl},\hat{\tau}_{D-1 l})^\T,\\
    	\boldsymbol{W}_l &= (\hat{\pi}_{2l},\hat{\tau}_{1l},\dots,\hat{\pi}_{tl},\hat{\tau}_{t-1 l})^\T.
    	\end{align*}
    	Denote \(\text{dim}(\boldsymbol{T}_l)=\left|T_l\right|\), and similarly for \(\boldsymbol{T}_l'\) and
    	\(\boldsymbol{W}_l\).
    	
    	By our assumptions, we can write
    	\[ \boldsymbol{A}^{-1} = \begin{pmatrix} \boldsymbol{A}_{\boldsymbol{T}_l\boldsymbol{T}_l}^{-1} & \boldsymbol{A}_{\boldsymbol{T}_l\boldsymbol{T}_l'}^{-1} \\[0.3em]
    	\boldsymbol{A}_{\boldsymbol{T}_l'\boldsymbol{T}_l}^{-1} & \boldsymbol{A}_{\boldsymbol{T}_l'\boldsymbol{T}_l'}^{-1}
    	\end{pmatrix}, \]
    	where for example \(\boldsymbol{A}_{\boldsymbol{T}_l\boldsymbol{T}_l}^{-1}=\text{cov}(\boldsymbol{T}_l,\boldsymbol{T}_l \mid L_{l1},\dots,L_{lD})\), and with \(\boldsymbol{A}_{T_lT_l}^{-1}\) and \(\boldsymbol{A}_{\boldsymbol{T}_l'\boldsymbol{T}_l}^{-1}=(\boldsymbol{A}_{\boldsymbol{T}_l\boldsymbol{T}_l'}^{-1})^\T\) holding the required conditions for the fixed effects. Finally, as part of our
    	inductive hypothesis we also assume that \(\text{det}⁡(\boldsymbol{A}_{\boldsymbol{T}_l\boldsymbol{T}_l}^{-1})>0\).
    	
    	With these definitions, our aim can then be stated to prove that
    	\[ \text{cov}(\boldsymbol{W}_{l+1},\boldsymbol{W}_{l+1} \mid L_{l+11},\dots,L_{l+1D})=\frac{l}{l+1}\boldsymbol{A}_{\boldsymbol{W}_l\boldsymbol{W}_l}^{-1},\]
    	i.e. that the covariance between the fixed effects
    	\(\hat{\pi}_{2l},\dots,\hat{\pi}_{tl}\) and
    	\(\hat{\tau}_{1l},\dots,\hat{\tau}_{t-1 l}\) is \(l/(l+1)\) that of its form at interim analysis \(l\), and similarly that
    	\[ \text{cov}(\boldsymbol{W}_{l+1}',\boldsymbol{W}_{l+1} \mid L_{l+11},\dots,L_{l+1D})=\frac{l}{l+1}\boldsymbol{A}_{\boldsymbol{W}_l'\boldsymbol{W}_l}^{-1}.\]
    	For brevity, from here we will write
    	\[ \text{cov}(\hat{\boldsymbol{\beta}}_l,\hat{\boldsymbol{\beta}}_l\mid L_{l1},\dots,L_{lD}) = \text{cov}(\hat{\boldsymbol{\beta}}_l,\hat{\boldsymbol{\beta}}_l),\]
    	and similarly for \(\boldsymbol{T}_l\), \(\boldsymbol{T}_l'\) and
    	\(\boldsymbol{W}_l\), for any $l$.
    	
    	We use the following identity, which requires only the invertibility of \(\boldsymbol{A}\) to
    	be valid (Henderson and Searle, 1981)
    	\[ (\boldsymbol{A}+\boldsymbol{U}\boldsymbol{C}\boldsymbol{V})^{-1}=\boldsymbol{A}^{-1}\left\{\boldsymbol{I}-\boldsymbol{U}\boldsymbol{C}\boldsymbol{V}\boldsymbol{A}^{-1}(\boldsymbol{I}+\boldsymbol{U}\boldsymbol{C}\boldsymbol{V}\boldsymbol{A}^{-1})^{-1}\right\}. \]
    	Note that we can write
    	\[ \boldsymbol{B} = \begin{pmatrix} \boldsymbol{B}_{\boldsymbol{T}_{l+1}\boldsymbol{T}_{l+1}} & \boldsymbol{0} \\[0.3em] \boldsymbol{0} & \boldsymbol{0} \end{pmatrix} = \begin{pmatrix} \boldsymbol{I}_{|\boldsymbol{T}_{l+1}|} \\[0.3em] \boldsymbol{0} \end{pmatrix} \boldsymbol{B}_{\boldsymbol{T}_{l+1}\boldsymbol{T}_{l+1}} \begin{pmatrix} \boldsymbol{I}_{|\boldsymbol{T}_{l+1}|} & \boldsymbol{0} \end{pmatrix} = \boldsymbol{U}\boldsymbol{C}\boldsymbol{V}, \]
    	since our general form for \( \sum_{i=1}^{|S_t|} \boldsymbol{X}_{s_{ti}}^\T \boldsymbol{\Sigma}_t^{-1}
    	\boldsymbol{X}_{s_{ti}} \) is only non-zero in a \(|\boldsymbol{T}_l|\times |\boldsymbol{T}_l|\) block in the top
    	left hand corner by Lemma~\ref{lemma2}. Therefore, provided \(\boldsymbol{A}\) is invertible, we can
    	always invert \(\boldsymbol{A}+\boldsymbol{B}\) to find the covariance matrix at the following interim
    	analysis. Moreover, we have
    	\begin{align*}
    	\boldsymbol{B}\boldsymbol{A}^{-1} &= \begin{pmatrix} \boldsymbol{B}_{\boldsymbol{T}_{l+1}\boldsymbol{T}_{l+1}} & 0 \\[0.3em] \boldsymbol{0} & \boldsymbol{0} \end{pmatrix} \begin{pmatrix} \boldsymbol{A}_{\boldsymbol{T}_l\boldsymbol{T}_l}^{-1} & \boldsymbol{A}_{\boldsymbol{T}_l\boldsymbol{T}_l'}^{-1} \\[0.3em]
    	\boldsymbol{A}_{\boldsymbol{T}_l'\boldsymbol{T}_l}^{-1} & \boldsymbol{A}_{\boldsymbol{T}_l'\boldsymbol{T}_l'}^{-1}
    	\end{pmatrix},\\ 
    	&= \begin{pmatrix} \boldsymbol{B}_{\boldsymbol{T}_{l+1}\boldsymbol{T}_{l+1}}\boldsymbol{A}_{\boldsymbol{T}_l\boldsymbol{T}_l}^{-1} & \boldsymbol{B}_{\boldsymbol{T}_{l+1}\boldsymbol{T}_{l+1}}\boldsymbol{A}_{\boldsymbol{T}_l\boldsymbol{T}_l'}^{-1} \\[0.3em] \boldsymbol{0} & \boldsymbol{0} \end{pmatrix},
    	\end{align*}
    	and
    	\[ \boldsymbol{I}_{2D-1} + \boldsymbol{B}\boldsymbol{A}^{-1} = \begin{pmatrix} \boldsymbol{B}_{\boldsymbol{T}_{l+1}\boldsymbol{T}_{l+1}}\boldsymbol{A}_{\boldsymbol{T}_l\boldsymbol{T}_l}^{-1} + \boldsymbol{I}_{|\boldsymbol{T}_l|} & \boldsymbol{B}_{\boldsymbol{T}_{l+1}\boldsymbol{T}_{l+1}}\boldsymbol{A}_{\boldsymbol{T}_l\boldsymbol{T}_l'}^{-1} \\[0.3em] \boldsymbol{0} & \boldsymbol{I}_{|\boldsymbol{T}_l'|} \end{pmatrix}, \]
    	
    	Now we need the formula (Henderson and Searle, 1981)
    	\begin{align*}
    	\boldsymbol{M} &= \begin{pmatrix} \boldsymbol{E} & \boldsymbol{F} \\[0.3em] \boldsymbol{G} & \boldsymbol{H} \end{pmatrix},\\ \Rightarrow \boldsymbol{M}^{-1} &= \begin{pmatrix} (\boldsymbol{E} - \boldsymbol{F}\boldsymbol{H}^{-1}\boldsymbol{G})^{-1} & -\boldsymbol{E}^{-1}\boldsymbol{F}(\boldsymbol{H}-\boldsymbol{G}\boldsymbol{E}^{-1}\boldsymbol{F})^{-1} \\[0.3em] -\boldsymbol{H}^{-1}\boldsymbol{G}(\boldsymbol{E}-\boldsymbol{F}\boldsymbol{H}^{-1}\boldsymbol{G})^{-1} & (\boldsymbol{H}-\boldsymbol{G}\boldsymbol{E}^{-1}\boldsymbol{F})^{-1} \end{pmatrix},
    	\end{align*}
    	which implies
    	
    	\begin{tiny}
    		\[ (\boldsymbol{I}_{2D-1} + \boldsymbol{B}\boldsymbol{A}^{-1})^{-1} = \begin{pmatrix} (\boldsymbol{B}_{\boldsymbol{T}_{l+1}\boldsymbol{T}_{l+1}}\boldsymbol{A}_{\boldsymbol{T}_l\boldsymbol{T}_l}^{-1} + \boldsymbol{I}_{|\boldsymbol{T}_l|})^{-1} & -(\boldsymbol{B}_{\boldsymbol{T}_{l+1}\boldsymbol{T}_{l+1}}\boldsymbol{A}_{\boldsymbol{T}_l\boldsymbol{T}_l}^{-1} + \boldsymbol{I}_{|\boldsymbol{T}_l|})^{-1}\boldsymbol{B}_{\boldsymbol{T}_{l+1}\boldsymbol{T}_{l+1}}\boldsymbol{A}_{\boldsymbol{T}_l\boldsymbol{T}_l'}^{-1} \\[0.3em] \boldsymbol{0} & \boldsymbol{I}_{|\boldsymbol{T}_l'|} \end{pmatrix}. \]
    	\end{tiny}
    	Note that to use this block matrix inversion formula we require
    	\(\boldsymbol{B}_{\boldsymbol{T}_{l+1}\boldsymbol{T}_{l+1}}\boldsymbol{A}_{\boldsymbol{T}_l\boldsymbol{T}_l}^{-1} + \boldsymbol{I}_{|\boldsymbol{T}_l|}\) to be invertible. However, by the previous result of this theorem only the variance of the intercept term in the form for
    	\(\boldsymbol{A}_{\boldsymbol{T}_l\boldsymbol{T}_l}^{-1}\) is dependent upon the value of \(t\), so we have that
    	\[ \boldsymbol{B}_{\boldsymbol{T}_{l+1}\boldsymbol{T}_{l+1}}\boldsymbol{A}_{\boldsymbol{T}_l\boldsymbol{T}_l}^{-1} = \frac{1}{l} \begin{pmatrix} V_1 & 0 \\[0.3em] \boldsymbol{V}_2 & \boldsymbol{I}_{|\boldsymbol{T}_l|-1} \end{pmatrix} = \frac{1}{l} \begin{pmatrix} V_1 & 0 \\[0.3em] \boldsymbol{V}_2 & \boldsymbol{I}_{|\boldsymbol{W}_l|} \end{pmatrix}, \]
    	for some \(V_1,\boldsymbol{V}_2\); \(\text{dim}(V_1)=1 \times 1, \text{dim}(\boldsymbol{V}_2)=|\boldsymbol{W}_l|
    	\times 1 \). Now, by Lemma~\ref{lemma3} \(\text{det}(\boldsymbol{B}_{\boldsymbol{T}_{l+1}\boldsymbol{T}_{l+1}}^{-1})>0\), which gives
    	\( \text{det}⁡(\boldsymbol{B}_{\boldsymbol{T}_{l+1}\boldsymbol{T}_{l+1}}) = \{ \text{det}⁡(\boldsymbol{B}_{\boldsymbol{T}_{l+1}\boldsymbol{T}_{l+1}}^{-1}) \}^{-1} >
    	0\). By assumption \(\text{det}⁡(\boldsymbol{A}_{\boldsymbol{T}_l\boldsymbol{T}_l}^{-1})>0\), and thus
    	\begin{align*}
    	\text{det}(\boldsymbol{B}_{\boldsymbol{T}_{l+1}\boldsymbol{T}_{l+1}}\boldsymbol{A}_{\boldsymbol{T}_l\boldsymbol{T}_l}^{-1}) &= \text{det}(\boldsymbol{B}_{\boldsymbol{T}_{l+1}\boldsymbol{T}_{l+1}})\text{det}(\boldsymbol{A}_{\boldsymbol{T}_l\boldsymbol{T}_l}^{-1})=\left(\frac{1}{l}\right)^{|\boldsymbol{W}_l|}V_1>0, \\
    	\Rightarrow V_1 &> 0.
    	\end{align*}
    	Therefore
    	\begin{align*}
    	\boldsymbol{B}_{\boldsymbol{T}_{l+1}\boldsymbol{T}_{l+1}}\boldsymbol{A}_{\boldsymbol{T}_l\boldsymbol{T}_l}^{-1} + \boldsymbol{I}_{|\boldsymbol{T}_l|} &= \frac{1}{l} \begin{pmatrix} V_1+l & 0 \\[0.3em] \boldsymbol{V}_2 & (1 + l)\boldsymbol{I}_{|\boldsymbol{W}_l|} \end{pmatrix},\\
    	\Rightarrow \text{det}(\boldsymbol{B}_{\boldsymbol{T}_{l+1}\boldsymbol{T}_{l+1}}\boldsymbol{A}_{\boldsymbol{T}_l\boldsymbol{T}_l}^{-1} + \boldsymbol{I}_{|\boldsymbol{T}_l|}) &> 0,
    	\end{align*}
    	and \( \boldsymbol{B}_{\boldsymbol{T}_{l+1}\boldsymbol{T}_{l+1}}\boldsymbol{A}_{\boldsymbol{T}_l\boldsymbol{T}_l}^{-1} + \boldsymbol{I}_{|\boldsymbol{T}_l|} \)  is therefore invertible as required.
    	
    	Now
    	\begin{tiny}
    		\begin{equation*}
    		\begin{split}
    		\text{cov}(\hat{\boldsymbol{\beta}}_{l+1},\hat{\boldsymbol{\beta}}_{l+1}) &= \begin{pmatrix} A_{\boldsymbol{T}_l\boldsymbol{T}_l}^{-1} & A_{\boldsymbol{T}_l\boldsymbol{T}_l'}^{-1} \\[0.3em]
    		A_{\boldsymbol{T}_l'\boldsymbol{T}_l}^{-1} & A_{\boldsymbol{T}_l'\boldsymbol{T}_l'}^{-1}
    		\end{pmatrix} \times \\
    		& \quad \times \left\{ I_{2D-1} - \begin{pmatrix} B_{\boldsymbol{T}_{l+1}\boldsymbol{T}_{l+1}}A_{\boldsymbol{T}_l\boldsymbol{T}_l}^{-1} & B_{\boldsymbol{T}_{l+1}\boldsymbol{T}_{l+1}}A_{\boldsymbol{T}_l\boldsymbol{T}_l'}^{-1} \\[0.3em] 0 & 0 \end{pmatrix} \right. \\
    		& \qquad \qquad \left. \begin{pmatrix} (B_{\boldsymbol{T}_{l+1}\boldsymbol{T}_{l+1}}A_{\boldsymbol{T}_l\boldsymbol{T}_l}^{-1} + I_{|\boldsymbol{T}_l|})^{-1} & -(B_{\boldsymbol{T}_{l+1}\boldsymbol{T}_{l+1}}A_{\boldsymbol{T}_l\boldsymbol{T}_l}^{-1} + I_{|\boldsymbol{T}_l|})^{-1}B_{\boldsymbol{T}_{l+1}\boldsymbol{T}_{l+1}}A_{\boldsymbol{T}_l\boldsymbol{T}_l'}^{-1} \\[0.3em] 0 & I_{|\boldsymbol{T}_l'|} \end{pmatrix}\right\}.
    		\end{split}
    		\end{equation*}
    	\end{tiny}
    	We thus have
    	\begin{small}
    		\begin{align*}
    		\text{cov}(\boldsymbol{T}_{l+1},\boldsymbol{T}_{l+1}) &= \boldsymbol{A}_{\boldsymbol{T}_l\boldsymbol{T}_l}^{-1}\left\{ \boldsymbol{I}_{|\boldsymbol{T}_l|} - \boldsymbol{B}_{\boldsymbol{T}_{l+1}\boldsymbol{T}_{l+1}}\boldsymbol{A}_{\boldsymbol{T}_l\boldsymbol{T}_l}^{-1}(\boldsymbol{B}_{\boldsymbol{T}_{l+1}\boldsymbol{T}_{l+1}}\boldsymbol{A}_{\boldsymbol{T}_l\boldsymbol{T}_l}^{-1} + \boldsymbol{I}_{|\boldsymbol{T}_l|})^{-1} \right\},\\
    		&= \boldsymbol{A}_{\boldsymbol{T}_l\boldsymbol{T}_l}^{-1}\left[\boldsymbol{I}_{|\boldsymbol{T}_l|} - \left\{\boldsymbol{I}_{|\boldsymbol{T}_l|} + (\boldsymbol{B}_{\boldsymbol{T}_{l+1}\boldsymbol{T}_{l+1}}\boldsymbol{A}_{\boldsymbol{T}_l\boldsymbol{T}_l}^{-1})^{-1} \right\}^{-1} \right],\\
    		\text{cov}(\boldsymbol{T}_{l+1}',\boldsymbol{T}_{l+1}) &= \boldsymbol{A}_{\boldsymbol{T}_l'\boldsymbol{T}_l}^{-1}\left\{ \boldsymbol{I}_{|\boldsymbol{T}_l|} - \boldsymbol{B}_{\boldsymbol{T}_{l+1}\boldsymbol{T}_{l+1}}\boldsymbol{A}_{\boldsymbol{T}_l\boldsymbol{T}_l}^{-1}(\boldsymbol{B}_{\boldsymbol{T}_{l+1}\boldsymbol{T}_{l+1}}\boldsymbol{A}_{\boldsymbol{T}_l\boldsymbol{T}_l}^{-1} + \boldsymbol{I}_{|\boldsymbol{T}_l|})^{-1} \right\},\\
    		&= \boldsymbol{A}_{\boldsymbol{T}_l'\boldsymbol{T}_l}^{-1}\left[\boldsymbol{I}_{|\boldsymbol{T}_l|} - \left\{\boldsymbol{I}_{|\boldsymbol{T}_l|} + (\boldsymbol{B}_{\boldsymbol{T}_{l+1}\boldsymbol{T}_{l+1}}\boldsymbol{A}_{\boldsymbol{T}_l\boldsymbol{T}_l}^{-1})^{-1} \right\}^{-1} \right],
    		\end{align*}
    	\end{small}
    	by the identity \(\boldsymbol{D}\boldsymbol{C}^{-1}(\boldsymbol{D}\boldsymbol{C}^{-1}+\boldsymbol{I})^{-1}=(\boldsymbol{I}+(\boldsymbol{D}\boldsymbol{C}^{-1})^{-1})^{-1}\)
    	(Henderson and Searle, 1981), which we can use as \(
    	\boldsymbol{B}_{\boldsymbol{T}_{l+1}\boldsymbol{T}_{l+1}}\boldsymbol{A}_{\boldsymbol{T}_l\boldsymbol{T}_l}^{-1} + \boldsymbol{I}_{|\boldsymbol{T}_l|} \)  is invertible from earlier.
    	Then
    	
    	\begin{tiny}
    		\begin{align*}
    		\text{cov}(\boldsymbol{T}_{l+1},\boldsymbol{T}_{l+1}) &= \boldsymbol{A}_{\boldsymbol{T}_l\boldsymbol{T}_l}^{-1}\left(\boldsymbol{I}_{|\boldsymbol{T}_l|} - \left[ \boldsymbol{I}_{|\boldsymbol{T}_l|} - \left\{\frac{1}{l} \begin{pmatrix} V_1 & 0 \\[0.3em] \boldsymbol{V}_2 & \boldsymbol{I}_{|\boldsymbol{W}_l|} \end{pmatrix} \right \}^{-1} \right]^{-1}  \right),\\
    		&= \boldsymbol{A}_{\boldsymbol{T}_l\boldsymbol{T}_l}^{-1}\left[\boldsymbol{I}_{|\boldsymbol{T}_l|} - \left\{ \boldsymbol{I}_{|\boldsymbol{T}_l|} - {l}\begin{pmatrix} V_1 & \boldsymbol{0} \\[0.3em] \boldsymbol{V}_2 & \boldsymbol{I}_{|\boldsymbol{W}_l|} \end{pmatrix}^{-1} \right\}^{-1} \right],\\
    		&= \boldsymbol{A}_{\boldsymbol{T}_l\boldsymbol{T}_l}^{-1}\left[\boldsymbol{I}_{|\boldsymbol{T}_l|} - \left\{ \boldsymbol{I}_{|\boldsymbol{T}_l|} - {l}\begin{pmatrix} V_1^{-1} & \boldsymbol{0} \\[0.3em] -\boldsymbol{I}_{|\boldsymbol{W}_l|}\boldsymbol{V}_2V_1^{-1} & \boldsymbol{I}_{|\boldsymbol{W}_l|}^{-1} \end{pmatrix} \right\}^{-1} \right],\\
    		&= \boldsymbol{A}_{\boldsymbol{T}_l\boldsymbol{T}_l}^{-1}\left\{\boldsymbol{I}_{|\boldsymbol{T}_l|} - \begin{pmatrix} 1 + lV_1^{-1} & \boldsymbol{0} \\[0.3em] -\boldsymbol{I}_{|\boldsymbol{W}_l|}\boldsymbol{V}_2V_1^{-1} & (1+l)\boldsymbol{I}_{|\boldsymbol{W}_l|}^{-1} \end{pmatrix}^{-1} \right\},\\
    		&= \boldsymbol{A}_{\boldsymbol{T}_l\boldsymbol{T}_l}^{-1}\times\\
    		& \qquad \times \left\{\boldsymbol{I}_{|\boldsymbol{T}_l|} - \begin{pmatrix} (1 + lV_1^{-1})^{-1} & \boldsymbol{0} \\[0.3em] -(1+l)^{-1}\boldsymbol{I}_{|\boldsymbol{W}_l|} \times -\boldsymbol{I}_{|\boldsymbol{W}_l|}\boldsymbol{V}_2V_1^{-1}(1+lV_1^{-1})^{-1} & (1+l)^{-1}\boldsymbol{I}_{|\boldsymbol{W}_l|} \end{pmatrix} \right\},\\
    		&= \boldsymbol{A}_{\boldsymbol{T}_l\boldsymbol{T}_l}^{-1}\left\{\boldsymbol{I}_{|\boldsymbol{T}_l|} - \begin{pmatrix} (1 + lV_1^{-1})^{-1} & \boldsymbol{0} \\[0.3em] (1+l)^{-1}\boldsymbol{I}_{|\boldsymbol{W}_l|}\boldsymbol{V}_2V_1^{-1}(1+lV_1^{-1})^{-1} & (1+l)^{-1}\boldsymbol{I}_{|\boldsymbol{W}_l|} \end{pmatrix} \right\},\\
    		&= \boldsymbol{A}_{\boldsymbol{T}_l\boldsymbol{T}_l}^{-1}\begin{pmatrix} 1 - (1 + lV_1^{-1})^{-1} & \boldsymbol{0} \\[0.3em] (1+l)^{-1}\boldsymbol{I}_{|\boldsymbol{W}_l|}\boldsymbol{V}_2V_1^{-1}(1+lV_1^{-1})^{-1} & \boldsymbol{I}_{|\boldsymbol{W}_l|}(1+l)^{-1}\boldsymbol{I}_{|\boldsymbol{W}_l|} \end{pmatrix}.
    		\end{align*}
    	\end{tiny}
    	Now
    	\begin{small}
    		\begin{align*}
    		\text{det}\{ \text{cov}(\boldsymbol{T}_{l+1},\boldsymbol{T}_{l+1} \mid L_{l+11},\dots,L_{l+1D})\} &= \text{det}(\boldsymbol{A}_{\boldsymbol{T}_l\boldsymbol{T}_l}^{-1})\text{det}\{1 - (1+lV_1^{-1})^{-1}\}\\
    		& \qquad \times \text{det}\{\boldsymbol{I}_{|\boldsymbol{W}_l| - (1+l)^{-1}\boldsymbol{I}_{|\boldsymbol{W}_l|}} \}\\
    		&>0,
    		\end{align*}
    	\end{small}
    	thus \( \text{det}(\boldsymbol{A}_{\boldsymbol{T}_{l+1}\boldsymbol{T}_{l+1}}^{-1})>0\) as required, and
    	\begin{align*}
    	\text{cov}(\boldsymbol{W}_{l+1},\boldsymbol{W}_{l+1})&=\boldsymbol{A}_{\boldsymbol{W}_l\boldsymbol{W}_l}^{-1}\{\boldsymbol{I}_{|\boldsymbol{W}_l|} - (1+l)^{-1}\boldsymbol{I}_{|\boldsymbol{W}_l|}\},\\&=\boldsymbol{A}_{\boldsymbol{W}_l\boldsymbol{W}_l}^{-1}\left\{\frac{l}{l+1}\boldsymbol{I}_{|\boldsymbol{W}_l|}\right\},\\&=\frac{l}{l+1}\boldsymbol{A}_{\boldsymbol{W}_l\boldsymbol{W}_l}^{-1}.
    	\end{align*}

    	Similarly
    	\begin{align*}
    	\text{cov}(\boldsymbol{W}_{l+1}',\boldsymbol{W}_{l+1}) &= \boldsymbol{A}_{\boldsymbol{W}'_l\boldsymbol{W}_l}^{-1}\{\boldsymbol{I}_{|\boldsymbol{W}_l|} - (1+l)^{-1}\boldsymbol{I}_{|\boldsymbol{W}_l|}\},\\&=\boldsymbol{A}_{\boldsymbol{W}'_l\boldsymbol{W}_l}^{-1}\left\{\frac{l}{l+1}\boldsymbol{I}_{|\boldsymbol{W}_l|}\right\},\\ &= \frac{l}{l+1}\boldsymbol{A}_{\boldsymbol{W}'_l\boldsymbol{W}_l}^{-1}.
    	\end{align*}
    	Thus the covariance of the fixed effects at the \((l+1)\)th analysis has the desired property.
    	
    	Now, as the base case consider having completed one stage of the trial, and
    	proceeding to complete another with any number of treatments \(t=2,\dots,D\)
    	remaining. Then, in this instance
    	\[ \boldsymbol{A}=\frac{n}{|S_D|}\sum_{i=1}^{|S_D|} \boldsymbol{X}_{s_{Di}}^\T \boldsymbol{\Sigma}_D^{-1} \boldsymbol{X}_{s_{Di}}. \]
    	By the previous result of this theorem this is indeed invertible and has the desired property, and moreover
    	by Lemma~\ref{lemma3} \(\text{det}⁡(\boldsymbol{A}_{\boldsymbol{T}_l\boldsymbol{T}_l}^{-1})>0\). The proof is then complete.
    \end{proof}
    
    \begin{lemma}
    	\label{lemma3} Consider the matrix from part (1) of Theorem~\ref{theorem1}; the case \(L_{l1}=\dots=L_{lD-1}=0,L_{ld}=l\)
    	\[ \left( \frac{ln}{\left|S_D\right|} \sum_{i=1}^{\left|S_D\right|} \boldsymbol{X}_{s_{Di}}^\T \boldsymbol{\Sigma}_D^{-1} \boldsymbol{X}_{s_{Di}} \right)^{-1}. \]
    	Now consider restricting to the columns and rows corresponding to
    	\[\boldsymbol{T}_l =
    	(\hat{\mu}_{0l}, \hat{\pi}_{2l}, \hat{\pi}_{3l}, \dots, \hat{\pi}_{tl},
    	\hat{\tau}_{1l}, \hat{\tau}_{2l}, \dots, \hat{\tau}_{t-1 l})^\T, \]
    	for some \(t=2,\dots,D\). The determinant of this matrix, \(\text{cov}(\boldsymbol{T}_l,\boldsymbol{T}_l \mid L_{l1}=\dots=L_{lD-1}=0,L_{lD}=l)\), is
    	strictly positive for any \(t\).
    \end{lemma}
    \begin{proof}
    	We have, by part (1) of Theorem~\ref{theorem1}
    	\[ \text{cov}(\boldsymbol{T}_l,\boldsymbol{T}_l \mid L_{l1}=\dots=L_{lD-1}=0,L_{lD}=l) = \frac{1}{ln} \begin{pmatrix}
    	\sigma_b^2 + \frac{2D-1}{D}\sigma_e^2 & \boldsymbol{M}^\T \\[0.3em] \boldsymbol{M} & \boldsymbol{N} \end{pmatrix}, \]
    	for
    	\begin{align*}
    	\boldsymbol{M} &= \begin{pmatrix} -\sigma_e^2 \\[0.3em] \vdots \\[0.3em] -\sigma_e^2 \end{pmatrix}, \\
    	\boldsymbol{N} &= \sigma_e^2 \begin{pmatrix} \boldsymbol{P} & \boldsymbol{0} \\[0.3em] \boldsymbol{0} & \boldsymbol{P} \end{pmatrix},\\
    	\boldsymbol{P} &= \begin{pmatrix} 2 & 1 & \dots & 1 \\[0.3em]
    	1 & \ddots & \ddots & \vdots \\[0.3em]
    	\vdots & \ddots & \ddots & 1 \\[0.3em]
    	1 & \dots & 1 & 2
    	\end{pmatrix},\\
    	\text{dim}(\boldsymbol{M}) &= (2t-2) \times 1,\\
    	\text{dim}(\boldsymbol{N}) &= 2(t-1) \times 2(t-1),\\
    	\text{dim}(\boldsymbol{P}) &= (t-1) \times (t-1).
    	\end{align*}
    	Then
    	\begin{scriptsize}
    		\begin{align*}
    		\text{det}\left\{\frac{1}{ln} \begin{pmatrix}
    		\sigma_b^2 + \frac{2D-1}{D}\sigma_e^2 & \boldsymbol{M}^\T \\[0.3em] \boldsymbol{M} & \boldsymbol{N} \end{pmatrix}\right\} &= \left(\frac{1}{ln}\right)^{2t-1} \text{det}\left(\begin{pmatrix}
    		\sigma_b^2 + \frac{2D-1}{D}\sigma_e^2 & \boldsymbol{M}^\T \\[0.3em] \boldsymbol{M} & \boldsymbol{N} \end{pmatrix}\right),\\
    		&= \left(\frac{1}{ln}\right)^{2t-1} \text{det}(\boldsymbol{N}) \text{det}\left\{\left(\sigma_b^2 + \frac{2D-1}{D}\sigma_e^2\right) - \boldsymbol{M}^\T \boldsymbol{N}^{-1}\boldsymbol{M}\right\}.
    		\end{align*}
    	\end{scriptsize}
    	Now
    	\begin{align*}
    	\boldsymbol{N}^{-1} &= \frac{1}{\sigma_e^2}\begin{pmatrix} \boldsymbol{P}^-1 & \boldsymbol{0} \\[0.3em] \boldsymbol{0} & \boldsymbol{P}^-1 \end{pmatrix},\\
    	\text{det}(\boldsymbol{N}) &= \sigma_e^{4(t-1)}\text{det}(\boldsymbol{P})^2.
    	\end{align*}
    	We are left therefore to find \(\boldsymbol{P}^{-1}\). We have \(\boldsymbol{P}\boldsymbol{Q}=\boldsymbol{I}_{D-1}\) for
    	\[ \boldsymbol{Q}=\frac{1}{t} \begin{pmatrix} t-1 & -1 & \dots & -1 \\[0.3em]
    	-1 & \ddots & \ddots & \vdots \\[0.3em]
    	\vdots & \ddots & \ddots & -1 \\[0.3em]
    	-1 & \dots & -1 & t-1
    	\end{pmatrix}. \]
    	Thus \(\boldsymbol{P}^{-1}=\boldsymbol{Q}\), and
    	\begin{tiny}
    		\begin{align*}
    		\text{det}\left\{\frac{1}{ln} \begin{pmatrix}
    		\sigma_b^2 + \frac{2D-1}{D}\sigma_e^2 & \boldsymbol{M}^\T \\[0.3em] \boldsymbol{M} & \boldsymbol{N} \end{pmatrix}\right\} &= \left(\frac{1}{ln}\right)^{2t-1} \text{det}\left(\begin{pmatrix}
    		\sigma_b^2 + \frac{2D-1}{D}\sigma_e^2 & \boldsymbol{M}^\T \\[0.3em] \boldsymbol{M} & \boldsymbol{N} \end{pmatrix}\right),\\
    		&= \left(\frac{1}{ln}\right)^{2t-1} \text{det}(\boldsymbol{N}) \text{det}\left\{\left(\sigma_b^2 + \frac{2D-1}{D}\sigma_e^2\right) - \boldsymbol{M}^\T \boldsymbol{N}^{-1}\boldsymbol{M}\right\},\\
    		&= \left(\frac{1}{ln}\right)^{2t-1} \sigma_e^{4(t-1)}\text{det}(\boldsymbol{P})^2\\ & \qquad \times \text{det}\left\{\left(\sigma_b^2 + \frac{2D-1}{D}\sigma_e^2\right) - \sigma_e^2 \begin{pmatrix} -1 \\[0.3em] \vdots \\[0.3em] -1 \end{pmatrix}^\T \begin{pmatrix} \boldsymbol{Q} & \boldsymbol{0} \\[0.3em] \boldsymbol{0} & \boldsymbol{Q} \end{pmatrix} \begin{pmatrix} -1 \\[0.3em] \vdots \\[0.3em] -1 \end{pmatrix} \right\}.
    		\end{align*}
    	\end{tiny}
    	This will be strictly positive provided
    	\[ \text{det}\left\{\left(\sigma_b^2 + \frac{2D-1}{D}\sigma_e^2\right) - \sigma_e^2 \begin{pmatrix} -1 \\[0.3em] \vdots \\[0.3em]
    	-1 \end{pmatrix}^\T \begin{pmatrix} \boldsymbol{Q} & \boldsymbol{0} \\[0.3em] \boldsymbol{0} & \boldsymbol{Q} \end{pmatrix}
    	\begin{pmatrix} -1 \\[0.3em] \vdots \\[0.3em] -1 \end{pmatrix} \right\} > 0. \]
    	But
    	\begin{tiny}
    		\begin{align*}
    		\text{det}\left\{\left(\sigma_b^2 +
    		\frac{2D-1}{D}\sigma_e^2\right) - \sigma_e^2 \begin{pmatrix} -1 \\[0.3em] \vdots \\[0.3em]
    		-1 \end{pmatrix}^\T \begin{pmatrix} \boldsymbol{Q} & \boldsymbol{0} \\[0.3em] \boldsymbol{0} & \boldsymbol{Q} \end{pmatrix}
    		\begin{pmatrix} -1 \\[0.3em] \vdots \\[0.3em] -1 \end{pmatrix} \right\} &=
    		\text{det}\left\{\left(\sigma_b^2 +
    		\frac{2D-1}{D}\sigma_e^2\right) - \sigma_e^2 \begin{pmatrix} -1/t \\[0.3em] \vdots \\[0.3em]
    		-1/t \end{pmatrix}^\T \begin{pmatrix} -1 \\[0.3em] \vdots \\[0.3em] -1 \end{pmatrix} \right\},\\
    		&= \text{det}\left\{ \left(\sigma_b^2 + \frac{2D-1}{D}\sigma_e^2\right) - \sigma_e^2\left(\frac{2t-2}{t}\right) \right\},\\
    		&= \sigma_b^2 + \sigma_e^2\left(\frac{2}{t} - \frac{1}{D} \right) > 0,
    		\end{align*}
    	\end{tiny}
    	since \(t \le D \). Thus, we have the required result.
    \end{proof}
    
    \section*{APPENDIX D: PROGRAMMES} \label{appD}
	
	\subsection{R}
	
	The \texttt{R} package \texttt{groupSeqCrossover} allows the determination,
	and exploration of, group sequential power family crossover trial designs.
	The function \texttt{gsco} is used to determine the design,
	taking inputs for the value of \(L\), \(\alpha\), \(\beta\), \(\delta\), $\sigma_e^2$, $\Delta$ and sequence type (\texttt{"latin"} or \texttt{"williams"}). The function \texttt{plot} can then be used through S3 methods to plot power curves, the expected sample size, and
	the expected number of observations, for varying true treatment effects. Moreover, the function \texttt{simulategsco} can be used to simulate group sequential crossover trials in order to assess their operating characteristics. This is especially useful in the case of small sample size designs. The code below for example identifies the discussed design for $\Delta=0$ and then plots the expected sample size curve
	
	\begin{verbatim}
	# Identify the design
	Delta.0 <- gsco(Delta = 0)
	# Plot E(N | tau_1 = ... = tau_(D-1) = theta)
	plot(Delta.0)
	\end{verbatim}
	
	Similarly, the following code would allow the determination of the familywise error rate under the global null hypothesis (when analysing using maximum likelihood estimation and using quantile substitution on the identified boundaries) of the design discussed in Appendix B
	
	\begin{verbatim}
	simulate.fwer <- simulategsco(REML = F, adjust = T)
	\end{verbatim}
	
	\subsection{Matlab}
	
	In order to ease the understanding of the forms of Equations (4.1) through (4.4),
	Matlab code employing symbolic algebra to return their forms is available. The user details a value for \(D\), and four matrices are then returned. For example, consider the case \(D = 4\)
	
	\begin{verbatim}
	>> [eq4_1, eq4_2, eq4_3, eq4_4] = groupSeqCrossoverMatrices(4);
	\end{verbatim}
	
	\texttt{eq4\_2} contains the \(\Sigma_r\) for \(r=2,3,4\). Specifically
	
	\begin{verbatim}
	>> eq4_2
	eq4_2 = [ b^2 + e^2,       b^2,       b^2,       b^2]
	[       b^2, b^2 + e^2,       b^2,       b^2]
	[       b^2,       b^2, b^2 + e^2,       b^2]
	[       b^2,       b^2,       b^2, b^2 + e^2]
	[ b^2 + e^2,       b^2,       b^2,         0]
	[       b^2, b^2 + e^2,       b^2,         0]
	[       b^2,       b^2, b^2 + e^2,         0]
	[         0,         0,         0,         0]
	[ b^2 + e^2,       b^2,         0,         0]
	[       b^2, b^2 + e^2,         0,         0]
	[         0,         0,         0,         0]
	[         0,         0,         0,         0]
	\end{verbatim}
	From this we observe
	\begin{align*}
	\boldsymbol{\Sigma}_2 &= \begin{pmatrix} \sigma_e^2 + \sigma_b^2 & \sigma_b^2 \\[0.3em]
	\sigma_b^2 & \sigma_e^2 + \sigma_b^2
	\end{pmatrix}, \\
	\boldsymbol{\Sigma}_3 &= \begin{pmatrix} \sigma_e^2 + \sigma_b^2 & \sigma_b^2 & \sigma_b^2 \\[0.3em]
	\sigma_b^2 & \sigma_e^2 + \sigma_b^2 & \sigma_b^2 \\[0.3em]
	\sigma_b^2 & \sigma_b^2 & \sigma_e^2 + \sigma_b^2
	\end{pmatrix}, \\
	\boldsymbol{\Sigma}_4 &= \begin{pmatrix} \sigma_e^2 + \sigma_b^2 & \sigma_b^2 & \sigma_b^2 & \sigma_b^2 \\[0.3em]
	\sigma_b^2 & \sigma_e^2 + \sigma_b^2 & \sigma_b^2 & \sigma_b^2 \\[0.3em]
	\sigma_b^2 & \sigma_b^2 & \sigma_e^2 + \sigma_b^2 & \sigma_b^2 \\[0.3em]
	\sigma_b^2 & \sigma_b^2 & \sigma_b^2 & \sigma_e^2 + \sigma_b^2
	\end{pmatrix}.
	\end{align*}
	Note that we have to remove the rows and columns of zeroes from the
	\(\boldsymbol{\Sigma}_r\) for \(r<D\), and we use b and e for \(\sigma_b\) and
	\(\sigma_e\) respectively.
	
	Similarly, \texttt{eq4\_1} contains the forms for \(\boldsymbol{\Sigma}_r^{-1}\) for \(r=2,3,4\).
	
	\texttt{eq4\_3} and \texttt{eq4\_4} correspond to the case
	\(\boldsymbol{\beta}=(\mu_0,\pi_2,\dots,\pi_D,\tau_1,\dots,\tau_{D-1})^\T\). \texttt{eq4\_3}
	contains
	\[ \sum_{i=1}^{|S_r|} \frac{n}{|S_r|} \boldsymbol{X}_{s_{ri}}^\T \boldsymbol{\Sigma}_r^{-1}\boldsymbol{X}_{s_{ri}} \quad (r=2,3,4).\]
	Precisely, the first \(2D-1\) rows correspond to the case \(r=4\), the next \(2D-1\) to \(r=3\), and so forth.
	
	Finally, \texttt{eq4\_4} contains the matrix \( \text{cov}(\hat{\boldsymbol{\beta}}_l, \hat{\boldsymbol{\beta}}_l \mid \boldsymbol{\omega}=(L,\dots,L)^\T,\boldsymbol{\psi}) \).
	
	\section*{ACKNOWLEDGEMENTS}
	This work was supported by the Wellcome Trust [grant number 099770/Z/12/Z to M.J.G.]; the Medical Research Council [grant number MC\_UP\_1302/2 to A.P.M.]; and the National Institute for Health Research Cambridge Biomedical Research Centre [MC\_UP\_1302/6 to J.M.S.W.].

\end{document}